\documentclass[12pt,a4paper]{article}
\usepackage[a4paper]{geometry}
\usepackage{epsfig}
\usepackage{amsmath, amsfonts, amstext, amscd, bezier,dsfont,bm,mathrsfs, bbm}
\usepackage{amssymb,amsthm}
\usepackage{indentfirst}
\usepackage{graphicx}
\usepackage{float}
\usepackage{color}
\usepackage{hyperref}
\usepackage[affil-it]{authblk}
\usepackage{subcaption} 
\usepackage{algorithm}
\usepackage{algpseudocode}
\usepackage{pifont}
\usepackage{algcompatible} 

\newcommand*{\Scale}[2][4]{\scalebox{#1}{$#2$}}

\usepackage{natbib}

\usepackage{mathtools}
\usepackage{seqsplit}

\usepackage{enumitem}
\newcommand{\subscript}[2]{$#1 _ #2$}

\theoremstyle{plain}
\newtheorem{theorem}{Theorem}
\newtheorem{lemma}[theorem]{Lemma}

\newtheorem{proposition}[theorem]{Proposition}

\theoremstyle{definition}

\theoremstyle{remark}

\providecommand{\keywords}[1]{\textit{Keywords:} #1}

\usepackage{mathtools}
\mathtoolsset{showonlyrefs}

\newcommand{\ind}{\mathds{1}}

\usepackage{amsmath, amssymb, amsfonts, amscd, xspace, pifont}

\newcommand{\V}[1]{\ensuremath{#1}\xspace}

\begin{document}
\title{\Large{A pseudo-likelihood approach to community detection  in weighted networks}}

\author[1]{\normalsize{Andressa Cerqueira}}
\author[2]{\normalsize{Elizaveta Levina}}

\affil[1]{\small{Department of Statistics, Federal University of S\~ao Carlos}}
\affil[2]{\small{Department of Statistics, University of Michigan}}

\date{\today}
\maketitle


\begin{abstract}
Community structure is common in many real networks, with nodes clustered in groups sharing the same connections patterns. While many community detection methods have been developed for networks with binary edges, few of them are applicable to networks with weighted edges, which are common in practice.    We propose a pseudo-likelihood community estimation algorithm derived under the weighted stochastic block model for networks with normally distributed edge weights, extending the pseudo-likelihood algorithm for binary networks, which offers some of the best combinations of accuracy and computational efficiency. We prove that the estimates obtained by the proposed method  are consistent under the assumption of homogeneous networks, a weighted analogue of the planted partition model, and show that they work well in practice for both homogeneous and heterogeneous networks.    We illustrate the method on simulated networks and on a fMRI dataset, where edge weights represent connectivity between brain regions and are expected to be close to normal in distribution by construction.  
\end{abstract}



\section{Introduction}

Network models have been a useful general tool for understanding and modeling interactions between objects in many domains.   The relationships (edges) between objects (nodes) can represent many things depending on the application: social interactions, trading partnerships, web links, packets sent between computers, disease contagion, neural connectivity, and so on.    Some settings result in binary networks, where only the presence or absence of an edge is recorded, and other settings lead to weighted networks, where an edge is associated with a weight which typically quantifies the strength of the connection.

The probabilistic modeling of networks has traditionally focused on binary networks, starting from the classical Erd\"os-R\'enyi graph \cite{erdHos1960evolution} which models edges as i.i.d.\ Bernoulli random variables.    Yet many networks encountered in practice are weighted, and many of the binary networks in the literature are obtained by thresholding raw edge weights.  For instance, the arguably most studied network with communities, the karate club dataset \cite{zachary1977information} is a binary network representing friendships, but the data were collected by observing frequency of the club members' interactions, and many other traditional social networks are recorded through similar mechanisms.   Moreover, even when there is a true binary edge (e.g., friendship on Facebook), there is often an accompanying strength measurement (e.g., how long these people have been friends, how many of each other's posts  they have liked, etc).

Network communities, a term used to loosely refer to groups of nodes sharing connectivity patterns, have been observed in many real-world networks and studied extensively. The stochastic block model (SBM) \cite{holland1983stochastic}, probably the most studied network model with communities, models a binary undirected graph with independent edges, with probability of an edge determined by community membership of the incident nodes.   A lot is now known about theory and algorithms for recovering communities in this setting;  see  \cite{abbe2017community} and references therein for a recent review.   There are also multiple extensions of the binary SBM, such as the degree-corrected SBM \citep{karrer2011stochastic}, and multiple models with overlapping communities, for example, \cite{airoldi2008mixed, zhang2014detecting, latouche2011overlapping}. A version called labeled SBM \citep{heimlicher2012community, lelarge2015reconstruction} allows for multiple different types of edges between a pair of nodes, which one could think of as an edge vector or a discrete categorical edge weight.

Our focus in this paper is on networks with real-valued edge weights.  They occur in many applications, notably brain connectivity networks obtained from various types of neuroimaging, and other biological networks such as gene-gene interactions.  Typically the weights represent some quantitative similarity measure between the nodes, such as marginal or partial correlations.   While these measures can be thresholded to obtain a binary network, there is at least empirical evidence that a lot of useful information is lost in the process \citep{relion2019network}.  

Several models have been proposed to directly model networks with continuous-valued edge weights.  For instance, \cite{aicher2014learning} proposed a generalization of the SBM for weighted edges and a variational Bayes approach to learning the latent community labels.     More recently, \cite{xu2017optimal} obtained the optimal error rate of any clustering algorithm for weighted SBM  through deriving an information-theoretic lower bound,  
and proposed an algorithm that achieves the optimal rate using discretization of weighted SBM into a labeled SBM.   Another line of work has focused on jointly analyzing multiple weighted networks, largely motivated by neuroimaging problems; see, for example, \cite{Levin2022recovering} and references therein.  
The work on community detection in multiple networks, however, has focused on binary networks to the best of our knowledge \citep{Arroyo2021inference, MacDonald2022latent, tang2009clustering, le2018estimating, bhattacharyya2018spectral, wang2019joint}.

Another related line of work is the submatrix localization problem, especially in the setting of  a random matrix with Gaussian entries with a constant variance.  The problem is to locate a submatrix of a known size of entries with positive means, while the rest of the matrix entries are assumed to have mean 0.   The case of one hidden submatrix can be viewed as a special case  of weighted SBM with two communities and was studied by \cite{ma2015computational} and \cite{hajek2017submatrix}.   \cite{chen2016statistical} studied the case of more than one hidden submatrix, all of the same size. 

In this paper, we consider the community detection problem for a single weighted SBM network.  We model edge weights as Gaussian random variables (in practice this may be after a suitable transformation, for example, the Fisher transform for correlations), with means and variances determined by the communities of the incident nodes.   Our main contribution is a fast and tractable  pseudo-likelihood algorithm for fitting this model, inspired by the pseudo-likelihood algorithm of \cite{amini2013pseudo} for binary networks, which remains one of the fastest and most accurate methods for community detection for both the binary SBM and the degree-corrected binary SBM.    We prove that one iteration step of our pseudo-likelihood algorithm achieves the optimal error rate derived by \cite{xu2017optimal}, up to a constant.    We characterize the optimal error rate for this setting explicitly  in terms of the means and variances of the within- and between-community edge distributions, revealing interesting trade-offs that do not appear in the binary setting.     We also show empirically that the pseudo-likelihood algorithm improves on other community detection methods when they are used to provide an initial value.     We allow for any fixed number of communities of different sizes, and do not assume assortativity in any form.  In contrast with the work \citep{aicher2014learning}, we also allow the means and variances of the edges weights to depend on the number of nodes, thus controlling the overall expected ``degree''.   

This paper is organized as follows. Section \ref{sec:wsbm}  introduces the weighted SBM and proposes the pseudo-likelihood based EM algorithm for fitting the model. Section \ref{sec:theo_results} presents our main consistency results. The performance of the algorithm on simulated networks and comparison with other methods are presented in Section \ref{sec:simulations}.   An application to brain connectivity networks is studied in Section \ref{sec:data_analysis}. We conclude with discussion in Section \ref{sec:discussion}.

\section{A Pseudo-Likelihood Algorithm for the Weighted Stochastic Block Model}\label{sec:wsbm}
Consider a weighted undirected network $W$ with the node set $\{1,2,\dots,n\}$.    Let $C_1,C_2,\dots, C_n$ be independent and identically distributed random variables representing node community labels,  with $\mathbb{P}(C_i=k)=\pi_k$, $k=1,\dots,K$ and $\pi=(\pi_1,\dots,\pi_K)$.  The \textit{weighted stochastic block model} (WSBM) assumes that, conditionally on the node labels $c = (c_1,\dots,c_n)$, the edge weights are drawn independently from a distribution that depends only on the labels of the incident nodes. In this paper, we study the model where the number of communities $K$ is fixed and known, and the distribution of edge weights is Gaussian.  Since our goal is to derive a pseudo-likelihood algorithm, some distributional assumption is needed, and many real weighted networks are reasonably well described by a Gaussian distribution of weights.  In particular, the motivating example in Section  \ref{sec:data_analysis} is on brain connectivity networks from fMRI, with edge weights measured as Fisher-transformed Pearson correlations, which are designed to be approximately Gaussian.  Formally, we have 
\begin{equation}
\mathscr{L} (W_{ij}\,|\,C_i=k, C_j = l) =  \mathscr{N}(B_{kl}\,,\,\Sigma_{kl})\,,\qquad 1\leq i < j \leq n\\
\end{equation}
where $B\in \mathbb{R}^{K\times K}$ and $\Sigma \in \mathbb{R}_{+}^{K\times K}$ are symmetric matrices that contain the means and the variances of the edge weights, respectively.    Diagonal elements of the matrix $W$ typically carry no information, and in some applications may be set to 0;  for completeness, we set the diagonal elements $W_{ii} = 0$, but their distribution can be left unspecified since they are not included in any calculations.  

Zero entries  in weighted networks are sometimes interpreted as non-observed edges, e.g., in  \cite{aicher2014learning} and \cite{xu2017optimal}.   We simply treat this as an edge with weight 0 (which has probability 0 in the Gaussian setting), and assume all edge weights are observed, resulting in a dense  weighted SBM.  This is a common scenario in real weighted networks, including those from neuroimaging.  

Under this model, the log-likelihood function for an observed network $W = w$ is given by
\begin{equation}\label{def:likelihood}
l(\pi,B,\Sigma;w) = \log\left( \sum\limits_{e\,\in\{1,\dots,K\}^n}p(w,e\,|\,\pi,B,\Sigma) \right)\,,
\end{equation}
where $e$ is an arbitrary community assignment, 
\begin{align}
  \label{def:joint}
  p(w,e\,|\,\pi, B,\Sigma) &  =    \prod_{k=1}^K \pi_k^{\,n_k(e)}  \times \\
                           & 
                             \prod_{k,l=1}^K\left[ 2\pi\Sigma_{kl} \right]^{- n_{kl}(e)/2}
                              \exp\left\{  -\sum_{1\leq i < j \leq n}\dfrac{(w_{ij}-B_{kl})^2}{2\Sigma_{kl}} \ind\{e_i=k,e_j=l\}\right\},                   
\end{align}
and 
\begin{equation}
n_k(\V{e}) = \sum\limits_{i=1}^n\ind\{e_i=k\} \ ,  \ n_{kl}(e) = \sum\limits_{ 1\leq i < j \leq n}\mathds{1}\{e_i=k,e_j=l\}\,
\end{equation}

Given an observed weighted network $W=w$, our goal is to estimate the label assignments $c$. Since the log-likelihood function \eqref{def:likelihood} involves a sum over all possible labels configurations ($K^n$ terms), maximizing the likelihood or applying the EM algorithm directly is not tractable, as discussed for the binary networks case in \cite{amini2013pseudo}.   \cite{aicher2014learning} proposed a variational Bayes approach for the case of edge weights drawn from exponential family distributions and demonstrated its good empirical performance, but variational Bayes tends to scale poorly with the number of nodes and there are no theoretical performance guarantees available for this method.  Another proposal is the discretization-based rate-optimal method by \cite{xu2017optimal}, which works by converting a weighted network to a labeled network and by applying spectral methods to labeled networks and does not make a distributional assumption about edge weights.  This method achieves the optimal error rate for a discretization level that satisfies theoretical requirements. However, in practice the performance of this method depends significantly on the choice of the level of discretization. 



Inspired by the pseudo-likelihood approach of \cite{amini2013pseudo} for binary networks, we aim to replace the full intractable likelihood function \eqref{def:likelihood} with an approximate likelihood of appropriate sums of edge weights.   
Let $\V{e}=(e_1,\dots,e_n)\in \{1,\dots ,K\}^n$ be an arbitrary vector of node labels. For $i=1,\dots,n$ and $k=1,\dots,K$, we define $S_{ik}(e)$ as the sum of weights of edges between node $i$ and all nodes in community $k$, that is, 
\begin{equation}\label{def:block_sum}
S_{ik}(e) = \sum\limits_{j=1}^nW_{ij}\ind\{e_j=k\}\, , 
\end{equation}
and write $s_{ik}$ for the corresponding observed quantity.   For each node $i$, define the vector of block sums $\V{S}_i(e)=(S_{i1}(e),\dots,S_{iK}(e))$.  
Let $R$ be the $K\times K$ confusion matrix between the label vector $\V{e}$ and the true labels $\V{c}$, defined by 
\begin{equation}
R_{kl} = \dfrac{1}{n}\sum\limits_{i=1}^n\ind\{e_i=k,c_i=l\}\,.
\end{equation}
Conditionally on the true labels $\V{c}$, $\{S_{i1}(e),\dots, S_{iK}(e)\}$ are mutually independent random variables. 
Moreover, again  conditionally on $c$, for a node $i$ with $c_i=k$, each variable $S_{il}$, $l = 1, \dots, K$, follows the normal distibtuion 
$\mathscr{N}(P_{kl}\,,\Lambda_{kl})$, 
where the mean and variance are given by
\begin{align*}
  P_{kl} & =nR_{l}.B._{k} \, , \ \Lambda_{kl}=nR_{l}.\Sigma._{k}\, , 
\end{align*}
and $M_{k}.$ and $M._{k}$ denote, respectively, the $k$th row and column of a matrix $M$.

The form of $K \times K$ parameter matrices $P$ and $\Lambda$ suggests a Gaussian mixture distribution, and indeed each random vector $\V{S}_i(e)$ can be written as a mixture of $K$  Gaussian random vectors, with the probability distribution function 
\begin{equation}
p(\,\V{s}_i(e)\, ; \,\pi,P,\Lambda) = \sum\limits_{l=1}^K\,\pi_l\,p(\,\V{s}_i(e)\, ; \,P_{l}.\,,\Lambda_{l}.)\,.
\end{equation}
%
Since we consider undirected networks, the matrix $W$ is symmetric, and the block sums for nodes $i$ and $j$ are not independent.   Specifically, they share exactly one common summand, $W_{ij} = W_{ji}$, while all other terms in the sums are independent by assumption.   The pseudo-likelihood part of our approach is to ignore this fairly weak dependence, and treat the row block sums as independent, leading to the 
pseudo log-likelihood function  (up to a constant) given by 
\begin{equation}\label{def:pseudo}
\ell_{\mathrm{PL}}(\pi,P,\Lambda;\{\V{s}_i(e)\}) = \sum\limits_{i=1}^n\log\left(  \sum\limits_{l=1}^K\,\pi_l \prod\limits_{k=1}^K\dfrac{1}{\sqrt{\Lambda_{lk}}}\exp\left\lbrace \dfrac{- (s_{ik}(e)-P_{lk})^2}{2\Lambda_{lk}} \right\rbrace \right)
\end{equation}
This function is the log-likelihood of a Gaussian mixture model, with mixture components labels matching the distribution of the latent communities, and thus fitting this model, which can be done by a standard EM algorithm for Gaussian mixture models, allows us to estimate the true labels $c=(c_1,\dots,c_n)$.  The EM algorithm starts from an initial value of the labels, say $c_0$, estimates the parameters given the labels, then updates the labels, and repeats this process either until convergence or for a fixed number of iterations $T$.    The steps of the EM algorithm are given in Algorithm \ref{alg:pseudo}.

\begin{algorithm}[H]
\caption{The pseudo-likelihood algorithm for estimating labels}
\label{alg:pseudo}
\begin{algorithmic}[]
\item[]\hspace{-0.7cm} \textbf{Input:} Initial labeling $c_0$, number of communities $K$ and the network matrix $W$ 
\item[]\hspace{-0.7cm} \textbf{Output:} The estimated label vector $\hat{c}$ \vspace{0.2cm}
\State 1. Initialize parameters $\widehat{\pi}_l( c_0)$, $\widehat{R}=\mathrm{diag}(\widehat{\pi}_1( c_0),\dots,\widehat{\pi}_K( c_0))$, $\widehat{P}_{lk}=n\widehat{R}_{k}.\widehat{B}._{l}$ and $\widehat{\Lambda}_{lk}=n\widehat{R}_{k}.\widehat{\Sigma}._{l}$.
\State 2. \textbf{Repeat} $T$ times
\State\hspace{1cm} 3. \textbf{Repeat} until the parameter estimates $\hat\pi$, $\hat P$ and $\hat \Lambda$ converge
\State\hspace{2cm} 3.1. Compute the block sums using \eqref{def:block_sum}.
\State\hspace{2cm} 3.2.  Estimate the probabilities $\mathbb{P}_{_{PL}}(c_i=l|\,\mathbf{s}_i(e)\,)$ by
\begin{equation}
\widehat{\pi}_{il} = \dfrac{\,\widehat{\pi}_l \prod\limits_{k=1}^K\dfrac{1}{(\widehat{\Lambda}_{lk})^{1/2}}\exp\left\lbrace \dfrac{- (s_{ik}(e)-\widehat{P}_{lk})^2}{2\widehat{\Lambda}_{lk}} \right\rbrace}{\sum\limits_{m=1}^K\,\widehat{\pi}_m \prod\limits_{k=1}^K\dfrac{1}{(\widehat{\Lambda}_{mk})^{1/2}}\exp\left\lbrace \dfrac{- (s_{ik}(e)-\widehat{P}_{mk})^2}{2\widehat{\Lambda}_{mk}} \right\rbrace}\,.
\end{equation}
\State\hspace{2cm} 3.3. Update the parameter values: 
\begin{equation}
\widehat{\pi}_l=\dfrac{1}{n}\sum\limits_{i=1}^n\widehat{\pi}_{il}\, , \  \widehat{P}_{lk}=\dfrac{\sum\limits_{i=1}^n\widehat{\pi}_{il}s_{ik}(e)}{\sum\limits_{i=1}^n\widehat{\pi}_{il}}\, , \ \widehat{\Lambda}_{lk}=\dfrac{\sum\limits_{i=1}^n\widehat{\pi}_{il}(s_{ik}(e)-\widehat{P}_{lk})^2}{\sum\limits_{i=1}^n\widehat{\pi}_{il}}\,.
\end{equation}

\State\hspace{1cm}\label{step_alg}  4. Update the labels: $e_i=\arg\max\limits_{l=1,\dots,K}\widehat{\pi}_{il}$. 
\item[]\hspace{-0.7cm} \textbf{Return:} $\hat c = e$
\end{algorithmic}
\end{algorithm}

Once we have estimated the labels $\hat c$, the parameters $\pi$, $B$ and $\Sigma$ can be easily obtained in closed form by maximizing the complete likelihood \eqref{def:joint}.  For any fixed label assignment $e$, the likelihood \eqref{def:joint} is maximized by 
\begin{align*}
\widehat{\pi}_k(e)& = \frac{n_k(e)}{n}\,, \\
\widehat{B}_{kl}(e) & = \dfrac{1}{n_{kl}(e)}\sum\limits_{1\leq i < j \leq n}w_{ij}\mathds{1}\{e_i=k,e_j=l\} \\
\widehat{\Sigma}_{kl}(e) & = \dfrac{1}{n_{kl}(e)}\sum\limits_{1\leq i < j \leq n}(w_{ij}-\widehat{B}_{kl})^2\,\mathds{1}\{e_i=k,e_j=l\}\,.
\end{align*}
Plugging in the labels found by Algorithm 1 gives the final parameter estimates.

As always, how well an EM algorithm performs depends on the initial value $c_0$.   The consistency analysis in Section \ref{sec:theo_results} quantifies this dependence in terms of the fraction of initial labels that match the true labels $c$.  In practice, the EM algorithm can be initialized with the output of any clustering algorithm, such as spectral clustering;  we discuss  different choices of the initial value in detail in Section \ref{sec:simulations}.
\section{Consistency results}\label{sec:theo_results}
We establish consistency of the estimated node labels $\hat{c}$ obtained from the pseudo-likelihood algorithm as the number of nodes $n$ grows.   We focus on what is called weak consistency in the literature: the fraction of mislabeled nodes converges to zero in probability as $n \rightarrow \infty$.   In this framework, we treat $c$ as an unknown parameter and define the error for an estimate $\hat{c}$ as
\begin{equation}\label{eq:def_loss}
L(\hat{c},c) = \min\limits_{\phi\in\Phi_K}\dfrac{1}{n} \sum\limits_{i=1}^n\ind\{ \hat{c}_i\neq \phi(c_i) \}\,,
\end{equation}
where $\Phi_K$ is the set of all permutations of community labels $\{1, \dots, K\}$.  Weak consistency requires that $\mathbb{P}(L(\hat c, c) > 0) \rightarrow 0$.  

We study consistency under the  \textit{homogeneous} weighted SBM, where all within-community edge weights have the same distribution, and so do all between-community edge weights.  This is a common theoretical framework, often called the planted partition model for binary networks.     In our setting, we assume that the  $K\times K$ mean matrix $B$ and the $K\times K$ variance matrix $\Sigma$ are given by
\begin{equation}\label{eq:model}
B_{kl}=
\begin{cases}
a, \text{ if } k=l\\
b, \text{ if } k\neq l
\end{cases} \qquad \text{and} \qquad
\Sigma_{kl}=
\sigma^2, \, \text{for all } k,l\,,
\end{equation}
where $a,b\in \mathbb{R}$ and $\sigma^2 >0$.  All of the parameters $a$, $b$ and $\sigma^2$ can vary with $n$, but most of the time we suppress this dependence to simplify notation.

Intuitively, the larger the absolute difference between the means $|a-b|$ is relative to $\sigma$, the easier it should be to recover communities. Besides that, as with all SBM-based community detection, the problem should get harder for larger $K$ and for unbalanced community sizes.  Yet even in the homogeneous setting, it is clear that the trade-offs for weighted SBM are more complex than they are for the binary SBM, not least because both the mean and the variance parameters are involved.  Our goal is to establish conditions for consistency that make these trade-offs as explicit as possible.

We will study a single label update step (4) of the algorithm, which updates the estimated labels  $\hat c^{(t)}$ to $\hat c^{(t+1)}$. 
Then for each node $i$, the label update is given by  
\begin{equation}\label{eq:EM_estimate}
\hat{c}^{(t+1)}_i 
 = \arg\max_{k=1, \dots,K} \left\lbrace \hat{\pi}_k(\hat{c}^{(t)}) - \left( \sum\limits_{m=1}^K  \dfrac{(s_{im}(\hat{c}^{(t)})-\hat{P}_{km}(\hat{c}^{(t)}))^2}{2\hat{\Lambda}_{km}(\hat{c}^{(t)})} + \dfrac{1}{2}\log\hat{\Lambda}_{km}(\hat{c}^{(t)})\right) \right\rbrace , 
\end{equation}
where $\hat{P}(\hat{c}^{(t)})=n(R(\hat{c}^{(t)})\hat{B}(\hat{c}^{(t)}))^T$, $\hat{\Lambda}(\hat{c}^{(t)})=n(R(\hat{c}^{(t)})\hat{\Sigma}(\hat{c}^{(t)}))^T$,  and $\hat{c}^{(0)}\equiv c_0$.

Since the estimator \eqref{eq:EM_estimate}  depends on $\hat{\pi}_k(\hat{c}^{(t)}),\hat{P}(\hat{c}^{(t)})$ and $\hat{\Lambda}(\hat{c}^{(t)})$ we would expect the consistency results to depend on how good these estimates are.  For the parameter values $a$ and $b$, we only require a very mild condition, a correct ordering of $\hat a(\hat{c}^{(t)})$ and $\hat b(\hat{c}^{(t)})$, and will assume the parameter estimates belong to the set
\begin{equation}\label{eq:assumption_parameters}
S_{a,b}=\{\, (\hat{a},\hat{b},\hat{\sigma}^2)\in \mathbb{R}^3\, |\, (\hat{a}-\hat{b})(a-b) > 0 \text{ and } \hat{\sigma}^2 >0 \,\}\,.
\end{equation}

For the initial labels $c_0$, we will express all the consistency results in terms of the proportion of labels in $c_0$ that match the true community labels $c$, up to a permutation of community labels $1, \dots, K$.    All consistency results on community detection involve this permutation which is luckily irrelevant in practice, and to keep notation simple, we will omit it from further discussion.   
The proportion of correct initial values is not relevant asymptotically, but our empirical results in Section \ref{sec:simulations} show, not surprisingly, that in practice starting from a good initial value is beneficial.  

\subsection{Balanced communities}

We start from the case of balanced communities, that is, we assume each community has  $m$ nodes and $n = mK$.
We further assume the initial labeling $c_0 \in \{1,\dots,K\}^n$ matches $\gamma m$ labels in each community, resulting in the overall matching proportion $\gamma \in (0,1)$.   The remaining misclassified $(1-\gamma)m$ nodes of each community are assumed to be distributed equally between the other $K-1$ community assignments.   This set of initial partitions can be formally written as 
\begin{align*}
  \mathcal{E}_{\gamma} = \big\{  & e \in \{1,\dots,K\}^n : \text{ for all } k, l = 1, \dots, K   \\ 
                                 &  \sum\limits_{i=1}^n\ind\{e_i=k,c_i=l\} =
                                   \gamma m \ind(k=l) + 
                                   \dfrac{(1-\gamma) m}{K-1} \ind (k\neq l) \big\}
\end{align*}
We emphasize that we do not know which initial labels are correct, only the proportion $\gamma$.

For any $\gamma\in(0,1)$, let $h(\gamma)=-\gamma\log(\gamma)-(1-\gamma)\log(1-\gamma)$ be the binary entropy and let $\kappa_{\gamma}(n) = \frac{1}{n}\left(\log n - \log \left( 4\pi\gamma(1-\gamma)+\frac{1}{3n}\right)\right)$.
The following theorem gives a probabilistic upper bound on the error of a single update step of the algorithm given by estimator \eqref{eq:EM_estimate}. Without loss of generality, we analyze the first step of the algorithm, that is, set $t=0$ and omit $t$ from now on.

\begin{theorem}[Balanced case]\label{theo:consitency}
Assume that $\pi_1=\dots=\pi_K=1/K$.  Let the initial labeling $e\in \mathcal{E}_{\gamma}$ and let $\hat{c}(e)$ be the estimate of the labels obtained from \eqref{eq:EM_estimate}.  
For $a,b\in \mathbb{R}$, $a\neq b$, $\sigma^2 > 0$, $\gamma\in (0,1)$, $\gamma\neq\frac{1}{K}$ we have 
\begin{equation}\label{eq:upper_bound_risk_bal}
\begin{split}
\sup\limits_{e\in \mathcal{E}_{\gamma}}\sup\limits_{\hat{a},\hat{b},\hat{\sigma}^2\,\in\, S_{a,b}}\,\mathbb{E}\left[\,L(\hat{c}\,(e),c)\,\right] \,\leq\, (K-1)\exp\left\lbrace -\dfrac{1}{4} \dfrac{(\gamma K-1)^2}{K(K-1)^2}\dfrac{n(a-b)^2}{\sigma^2}\right\rbrace\,,
\end{split}
\end{equation}
and
\begin{equation}\label{eq:Prob_upper_bound}
\begin{split}
&\mathbb{P}\left(\,\sup\limits_{e\in \mathcal{E}_{\gamma}}\,\sup\limits_{\hat{a},\,\hat{b},\,\hat{\sigma}^2\,\in\, S_{a,b}}L(\hat{c}\,(e),c) > \exp\left\lbrace -\dfrac{1}{8} \dfrac{(\gamma K-1)^2}{(K-1)^2}\dfrac{n}{K}\dfrac{(a-b)^2}{\sigma^2}\right\rbrace\,\right) \\
&\hspace{3,5cm} \leq ({K-1})\exp\left\lbrace -n\left(\dfrac{1}{8} \dfrac{(\gamma K-1)^2}{K(K-1)^2}\dfrac{(a-b)^2}{\sigma^2} - C(n,\gamma)\right)\right\rbrace
\end{split}
\end{equation}
where $C(n,\gamma) = h(\gamma) + \kappa_{\gamma}(2n/K)+(1-\gamma)\log(K-1\,)$.
As long as 
\begin{equation}\label{eq:condition_consistency}
\frac{1}{8} \frac{(\gamma K-1)^2}{K(K-1)^2}\dfrac{(a-b)^2}{\sigma^2} > C(n,\gamma) \, , 
\end{equation}
the pseudo-likelihood estimator is uniformly weakly consistent.
\end{theorem}

Theorem \ref{theo:consitency} holds both for $a>b$ and $a<b$.   Further, when parameter values scale with the number of nodes $n$,  the consistency result of Theorem \ref{theo:consitency} holds as long as 
\begin{equation}\label{eq:cond_vary_n}
 n\dfrac{(a_n-b_n)^2}{\sigma_n^2} \to \infty \qquad \text{when}\qquad n\to \infty\,.
\end{equation}

Condition \eqref{eq:condition_consistency} in Theorem \ref{theo:consitency} can be viewed as a requirement on the minimum difference between the distributions of within-community and between-community edges.     This is not a particularly strong requirement since the term $\kappa_{\gamma}(2n/K)$ in $C(n,\gamma)$ goes to zero as $n$ grows and the binary entropy term is bounded by $1$.   Empirical studies in Section \ref{sec:simulations} confirm the method works well even when the difference in the means is small.  

A recent result by \cite{xu2017optimal} established the optimal error rate for the weighted SBM, for any clustering algorithm, as a function of the Renyi divergence of order $1/2$ between the within- and between-community edge distributions.
We compute this bound explicitly for our model \eqref{eq:model}, with $N(a,\sigma^2)$ and $N(b,\sigma^2)$ distributions for the within- and between-community edges, respectively. In this case, the lower bound from \cite{xu2017optimal}, for the balanced communities case, is given by 
\begin{equation}
\exp\left\lbrace - (1+o(1)\,)\frac{n}{K}\dfrac{(a-b)^2}{4\sigma^2}\right\rbrace\, , 
\end{equation}
and our upper bound \eqref{eq:upper_bound_risk_bal} matches this theoretical lower bound up to a constant. Moreover, by \eqref{eq:Prob_upper_bound}
\begin{equation}
\lim\limits_{n\to \infty}\mathbb{P}\left(\,\sup\limits_{e\in \mathcal{E}_{\gamma}}\,\sup\limits_{\hat{a},\,\hat{b},\,\hat{\sigma}^2\,\in\, S_{a,b}}L(\hat{c}\,(e),c) > \exp\left\lbrace -\dfrac{1}{8} \dfrac{(\gamma K-1)^2}{(K-1)^2}\dfrac{n}{K}\dfrac{(a-b)^2}{\sigma^2}\right\rbrace\,\right) =0\,,
\end{equation}
when $\,\frac{1}{8} \frac{(\gamma K-1)^2}{K(K-1)^2}\frac{(a-b)^2}{\sigma^2} > C(n,\gamma)$. This implies that the pseudo-likelihood algorithm achieves the optimal rate up to a constant.

\subsection{Unbalanced communities}
The case of unbalanced communities is substantially more complicated, because now both the number of nodes and the proportion of correct initial labels in each community affect the performance.   To keep the technical details manageable and focus on understanding trade-offs, we limit our study of the unbalanced case to $K=2$.  

Let  $n_k = \sum\limits_{i=1}^n\ind\{c_i=k\}$ be the number of nodes in community $k$, $k = 1,2$, and let $\pi_k= n_k/n$.     Assume that the initial labeling $c_0\in \{1,2\}^n$, up to a permutation of labels,  matches $\gamma_1n_1$ labels in community $1$ and $\gamma_2n_2$ labels in community $2$, or in other words, the initial label vector belongs to the set 
\begin{equation}
\mathcal{E}_{\,\gamma_1,\gamma_2} = \left\lbrace e\in \{1,2\}^n \,:\, \sum\limits_{i=1}^n\ind\{e_i=k,c_i=k\} = \gamma_kn_k,\, k=1,2 \right\rbrace\,.
\end{equation}

For $e\in \mathcal{E}_{\,\gamma_1,\gamma_2}$, the confusion matrix $R$ is given by
\begin{equation}
R(e) = \left(\begin{array}{cc}
\gamma_1\dfrac{n_1}{n} & (1-\gamma_2)\dfrac{n_2}{n} \\ 
(1-\gamma_1)\dfrac{n_1}{n} & \gamma_2\dfrac{n_2}{n}
\end{array}\right) =  \left(\begin{array}{cc}
\gamma_1\pi_1 & (1-\gamma_2)\pi_2 \\ 
(1-\gamma_1)\pi_1 & \gamma_2\pi_2
\end{array}\right) . 
\end{equation}

For any $e\in \mathcal{E}_{\,\gamma_1,\gamma_2}$, observe that the number of nodes in each community can be expressed as 
\begin{eqnarray*}
\tilde{n}_1 &  =  & \sum\limits_{i=1}^n\ind\{e_i=1\} = \gamma_1n_1 + (1-\gamma_2)n_2 \, ,  \\
\tilde{n}_2 &  = & \sum\limits_{i=1}^n\ind\{e_i=2\} = (1-\gamma_1)n_1 + \gamma_2n_2\,.
\end{eqnarray*}
The corresponding proportions of nodes in each community of $e\in \mathcal{E}_{\,\gamma_1,\gamma_2}$ are then $\tilde{\pi}_1 = \gamma_1\pi_1 + \pi_2(1-\gamma_2)$ and $
\tilde{\pi}_2 = (1-\gamma_1)\pi_1 + \pi_2\gamma_2$.  
Conditional on the true labels $c$, the proportion of nodes in each community defined by $e$ depends only on the known matching proportions $\gamma_1$ and $\gamma_2$, and on the true proportions $\pi_1$ and $\pi_2$. Define the quantities
\begin{equation}\label{eq:unb_beta}
\begin{split}
&\beta_1= \tilde{\pi}_2\left((1-\gamma_2)\pi_2 - \gamma_1\pi_1\right) , \\
&\beta_2= \tilde{\pi}_1\left((1-\gamma_1)\pi_1-\gamma_2\pi_2\right) . 
\end{split}
\end{equation}
For parameter estimates $\hat{a},\hat{b},\hat{\sigma}^2\in \hat{P}_{a,b,\sigma^2}$, define the function
\begin{equation}
\begin{split}
F(x,y) \,=\, & (-2x+\hat{a}+\hat{b})\left(\beta_1\gamma_1\pi_1-\beta_2(1-\gamma_1)\pi_1\right) \\
& \qquad + (-2y+\hat{a}+\hat{b})\left(\beta_1(1-\gamma_2)\pi_2-\beta_2\gamma_2\pi_2\right) \, . 
\end{split}
\end{equation}By \eqref{eq:unb_beta}, the values of $\beta_1$ and $\beta_2$ depend only on $\gamma_1$, $\gamma_2$, $\pi_1$ and $\pi_2$, and we suppress this dependence to simplify  notation. Theorem \ref{prop:unb_misclassification} gives an upper bound on the probability of misclassification of each node $i$ in terms of $\gamma_1,\gamma_2,\hat{a},\hat{b},\hat{\sigma}^2,\pi_1$ and $\pi_2$.

\begin{theorem}[Unbalanced case]\label{prop:unb_misclassification}  Initialize the algorithm with a label vector $e\in\mathcal{E}_{\,\gamma_1,\gamma_2}$ and the corresponding parameter estimates $\hat{a}, \hat{b}, \hat{\sigma}^2\in S_{a,b}$.   Let 
\begin{equation}\label{eq:condition_unb}
\,t_1 = \dfrac{2\,\hat{\sigma}^2\tilde{\pi}_1\tilde{\pi}_2}{n(\hat{a}-\hat{b})}\log\left(\dfrac{\tilde{\pi}_1}{\tilde{\pi}_2}\right) + F(a,b) \, , \quad  t_2=\dfrac{\,2\hat{\sigma}^2\tilde{\pi}_1\tilde{\pi}_2}{n(\hat{a}-\hat{b})}\log\left(\dfrac{\tilde{\pi}_1}{\tilde{\pi}_2}\right) +F(b,a)\, . 
\end{equation}
If $a>b$, assume that $t_1\ge 0$ and $t_2< 0$, and  if $a < b$, assume that $t_1<0$ and $t_2\geq 0$.    Then the probability of misclassification for any node $i\in \{1,\dots,n\}$ is given by
\begin{align}
  \label{eq:unb_missclassification}
\mathbb{P}(\hat{c}_i(e)\neq 1 \,|\,c_i=1) & \leq \exp\left\lbrace -\dfrac{n}{2\sigma^2\tau^2}\,t_1^2\,\right\rbrace\,, \\
\label{eq:unb_missclassification2}
\mathbb{P}(\hat{c}_i(e)\neq 2\,|\,c_i=2) & \leq \exp\left\lbrace -\dfrac{n}{2\sigma^2\,\tau^2}\,t_2^2\,\right\rbrace\, , 
\end{align}
 where $\tau^2=\beta_1^2{\pi}_1 + \beta_2^2{\pi}_2$.  
\end{theorem}
Then \eqref{eq:condition_unb} implies that the probabilities of misclassification \eqref{eq:unb_missclassification} and \eqref{eq:unb_missclassification2} go to zero as $n$ grows.

To get more intuition about the tradeoffs indicated by Theorem \ref{prop:unb_misclassification}, take a $\delta_n > 0$ such that  $|\hat a_n - a_n | < \delta_n$, $|\hat b_n - b_n | < \delta_n$ and assume the variance is known, with $\hat{\sigma}^2_n=\sigma^2_n$.   We can then rewrite the bound \eqref{eq:unb_missclassification} in Theorem \ref{prop:unb_misclassification} as
\begin{align*}
  \exp  \left\lbrace
                      -   \frac{n(a_n-b_n)^2}{2\tau^2\sigma_n^2}  \times
           \right.   
    \left.  \left(
                           \frac{2 {\sigma_n}^2\tilde{\pi}_1\tilde{\pi}_2}{n({a_n}-{b_n+2\delta_n})({a_n}-{b_n})}
                            \log\left(\frac{\tilde{\pi}_1}{\tilde{\pi}_2}
                   \right)   
                  + C_1^{\gamma,\pi}
                  - \frac{2\delta C_2^{\gamma,\pi}}{(a_n-b_n)}
                   \right)^2
       \right\rbrace \,  ,
\end{align*}
where $C_1^{\gamma,\pi} = \beta_1\pi_2 + \beta_2\pi_1 - (\beta_1+\beta_2)(\pi_1\gamma_1+\pi_2\gamma_2)$ and $C_2^{\gamma,\pi} = \beta_1\pi_2 - \beta_2\pi_1 + (\beta_1+\beta_2)(\pi_1\gamma_1-\pi_2\gamma_2)$. If
\begin{equation*}
\dfrac{\sigma^2_n}{n} \to 0 \quad\mbox{ and }\quad \dfrac{n(a_n-b_n)^2}{\sigma_n^2} \to \infty\, ,
\end{equation*}
then the probabilities of misclassification \eqref{eq:unb_missclassification} and \eqref{eq:unb_missclassification2} go to zero as $n \rightarrow \infty$.

%

Combining the two bounds of Theorem \ref{prop:unb_misclassification}, we can conclude that the expected error of $\hat{c}(e)$, for any $e\in \mathcal{E}_{\gamma_1,\gamma_2}$, satisfies 
\begin{align}
  \label{eq:error_bound}
\mathbb{E}\left[L(\hat{c}(e),c) \right] \leq
  \pi_1\exp\left\lbrace -\dfrac{n}{2\sigma^2\tau^2}   t_1^2 \right\rbrace
  + \pi_2\exp\left\lbrace -\dfrac{n}{2\sigma^2 \tau^2} t_2^2 \right\rbrace
  \,.
\end{align}

In particular, when $\pi=(1/2,1/2)$ and $\gamma_1=\gamma_2=\gamma$, we have $\tilde \pi_1=\tilde \pi_2 = 1/2$, $\beta_1=\beta_2=\frac{1}{4}(1-2\gamma)$ and $-t_2=t_1=\frac{1}{4}(1-2\gamma)^2(a-b)$. Thus, the expected error given by \eqref{eq:error_bound} matches, up to a constant, the expected error \eqref{eq:upper_bound_risk_bal} obtained for the balanced case.  
  
Figure \ref{fig:condition_unb} shows the logarithm of the bound on the expected error as a function of the signal strength $|a-b|$, the difference between the within- and between-community means, and the quality of the initial parameter estimates $\delta$.  We fix  $n=100$ and the variance  $\hat \sigma^2 = 1$.   We plot the value of the log-bound as a heatmap, for the set of parameters that satisfy conditions \eqref{eq:condition_unb}.   As one would expect, the error decreases as $\delta$ decreases (better initial value for the parameters), as $|a-b|$ increases (easier problem), and as the proportion of matches $\gamma$ increases (better initial value for the labels).   We also see that a better initial value for the labels (higher $\gamma$) leads to a larger set of parameters satisfying the conditions, and to lower errors.   The effect of unbalanced communities is harder to isolate in this limited set of plots, though our empirical simulation results show that generally the errors for the unbalanced case tend to be higher than for the balanced case.  



\begin{figure}[H]
\centering
  \begin{subfigure}[b]{0.47\textwidth}
    \includegraphics[scale=0.35]{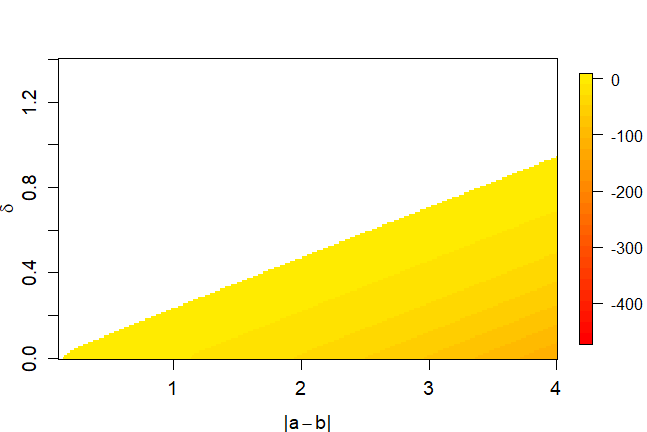}
    \caption{$\pi=(0.7,0.3)$, $\gamma_1=\gamma_2=0.6$}
  \end{subfigure}
  \begin{subfigure}[b]{0.47\textwidth}
    \includegraphics[scale=0.35]{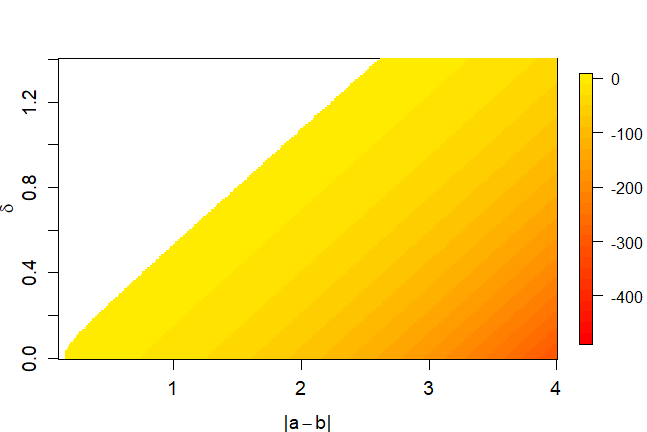}
    \caption{$\pi=(0.7,0.3)$, $\gamma_1=\gamma_2=0.8$}
  \end{subfigure}
  \begin{subfigure}[b]{0.47\textwidth}
    \includegraphics[scale=0.35]{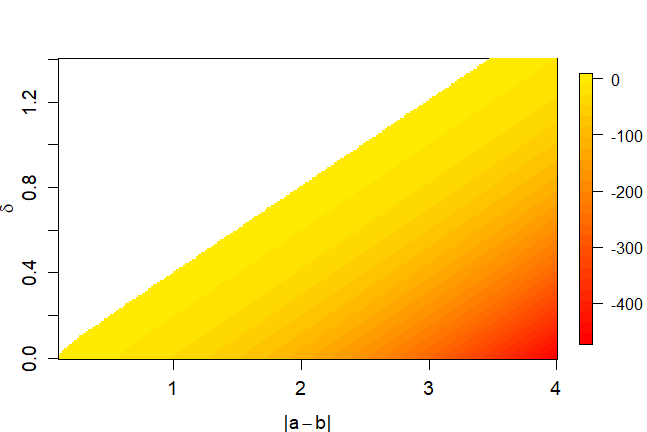}
    \caption{$\pi=(0.9,0.1)$, $\gamma_1=\gamma_2=0.6$}
  \end{subfigure}
  \begin{subfigure}[b]{0.47\textwidth}
    \includegraphics[scale=0.35]{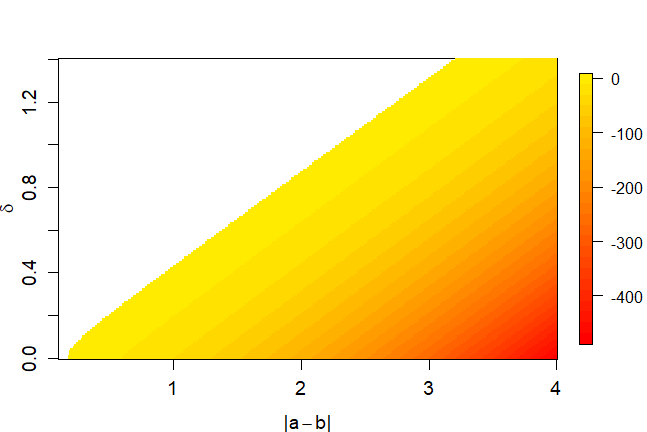}
    \caption{$\pi=(0.9,0.1)$, $\gamma_1=\gamma_2=0.8$}
  \end{subfigure}
  \begin{subfigure}[b]{0.47\textwidth}
    \includegraphics[scale=0.35]{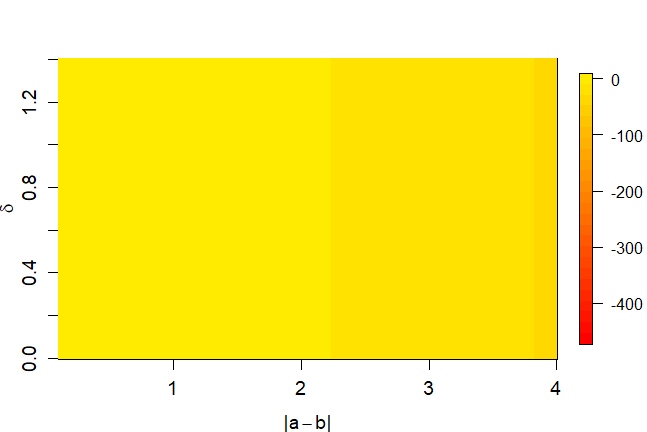}
    \caption{$\pi=(0.5,0.5)$, $\gamma_1=\gamma_2=0.6$}
  \end{subfigure}
  \begin{subfigure}[b]{0.47\textwidth}
    \includegraphics[scale=0.35]{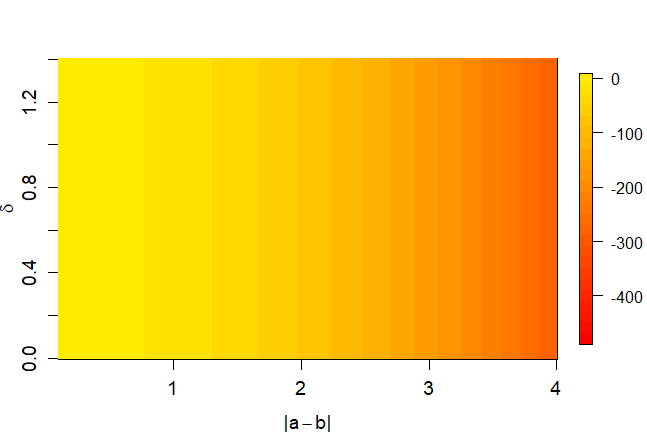}
    \caption{$\pi=(0.5,0.5)$, $\gamma_1=\gamma_2=0.8$}
  \end{subfigure}
  \caption{The logarithm of the upper bound of the expected error \eqref{eq:error_bound} when $n=100$, $\hat{\sigma}^2=\sigma^2=1$,  $|\hat a_n - a_n | < \delta_n$ and $|\hat b_n - b_n | < \delta_n$  for different proportion of matches $\gamma_1$ and $\gamma_2$. }
  \label{fig:condition_unb}
\end{figure}

\section{Empirical evaluation on simulated networks}\label{sec:simulations}

Our empirical investigations focus on two goals:  understanding how the various parameters of the problem affect the performance of  the pseudo-likelihood algorithm (PL), and comparing it to other ways of estimating communities from weighted networks.   We simulate networks from the model  \eqref{eq:model} with $K=3$ and other parameters as specified below.    The number of iterations of the PL algorithm is fixed at $20$.   Performance is evaluated by the the error in community assignments  defined in \eqref{eq:def_loss},  averaged over $100$ replications.  

Figure \ref{fig:pseudo_gamma} shows the performance of the PL algorithm, as implemented in Algorithm 1, as a function of the signal strength $|a-b|$, with fixed  $\sigma^2=1$ and several values of the number of nodes $n$ and the correct fraction of initial labels $\gamma$.  The results are intuitive:  the error rate decreases as the signal gets stronger (larger  $|a-b|$), the initial value improves  (larger $\gamma$), and the number of nodes grows.    An encouraging finding is that these results are not especially sensitive to $\gamma$, which we cannot easily control in practice.     When there is little difference between the means $a$ and $b$, the error rate gets close to random guessing (2/3 in this case, as $K=3$), as one would expect.     The comparison between balanced and unbalanced community sizes is not  straightforward with $K=3$, but overall the balanced case is easier, as one would expect.  The higher errors for the balanced case when the signal is very low are an artifact of the fact that putting all nodes into the single largest community gives the error rate of 66\% for the balanced case but only 50\% for the unbalanced case.   

\begin{figure}[!h]
\centering
  \begin{subfigure}[b]{0.48\textwidth}
    \includegraphics[scale=0.45]{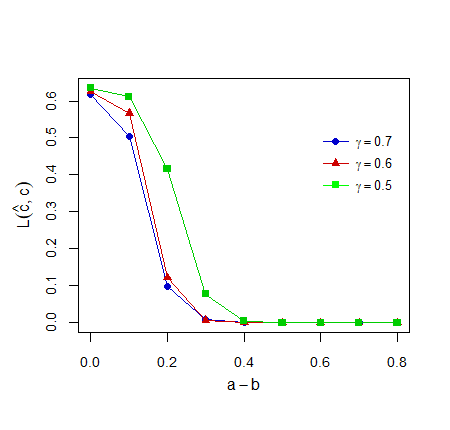}
    \caption{$n=500$ and $\pi=(1/3,1/3,1/3)$.}
    \label{fig:pseudo_gamma_a}
  \end{subfigure}
  \begin{subfigure}[b]{0.48\textwidth}
    \includegraphics[scale=0.45]{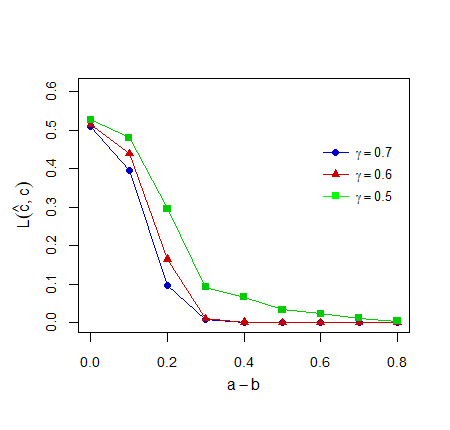}
    \caption{$n=500$ and $\pi=(0.2,0.5,0.3)$.}
    \label{fig:pseudo_gamma_b}
  \end{subfigure}
  \begin{subfigure}[b]{0.48\textwidth}
    \includegraphics[scale=0.45]{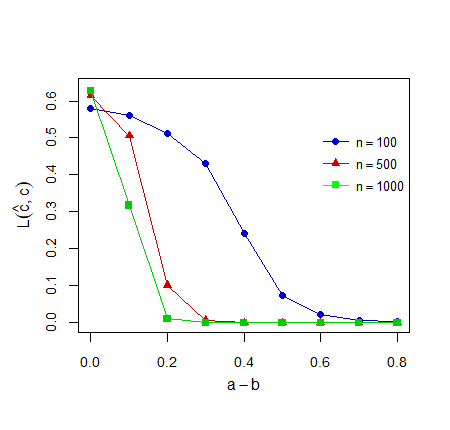}
    \caption{$\gamma=0.7$ and $\pi=(1/3,1/3,1/3)$.}
   \label{fig:pseudo_gamma_c}
  \end{subfigure}
  \begin{subfigure}[b]{0.48\textwidth}
    \includegraphics[scale=0.45]{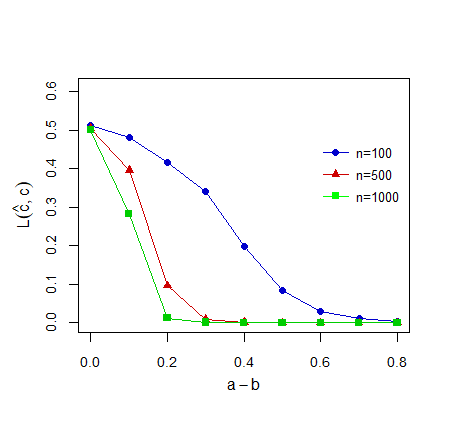}
    \caption{$\gamma=0.7$ and $\pi=(0.2,0.5,0.3)$.}
    \label{fig:pseudo_gamma_d}
  \end{subfigure}
  \caption{Error of the PL algorithm as a function of the difference between the means $|a-b|$, for several values of the correct proportion $\gamma$ and the number of nodes $n$, averaged over 100 replications.  The variance is fixed at $\sigma^2=1$.}
  \label{fig:pseudo_gamma}
\end{figure}

We compare two choices of the initial labels for the PL algorithm.  One is spectral clustering  (SC),  implemented as described in \cite{lei2015consistency} for unweighted networks, applied directly to the matrix of weights $W$.  The second option is the discretization-based algorithm (DB) of \cite{xu2017optimal}, which first discretizes the matrix of weights and then applies clustering.  The DB method itself depends on the choice of the discretization level; we followed the DB authors' recommendation and set it to $\lfloor 0.4(\log \log n )^4 \rfloor$.   

Figure \ref{fig:pseudo_initial_SC_DB} shows that the PL algorithm improves substantially upon both initial values.   The SC algorithm is generally more accurate than DB, and thus leads to better PL solutions when used as the initial value.   For comparison, the PL algorithm started with $\gamma = 0.7$ correct labels is included in all scenarios.    Spectral clustering is known to favor balanced solutions, and thus for unbalanced networks with a small community containing only $10\%$ of the nodes (Figure \ref{subfig:unb2}), the SC estimates are the worst among all the methods, but even then the PL method is able to improve somewhat on the initial value provided by spectral clustering.  

\begin{figure}[!htb]
\centering
  \begin{subfigure}[b]{0.48\textwidth}
  \centering
    \includegraphics[scale=0.45]{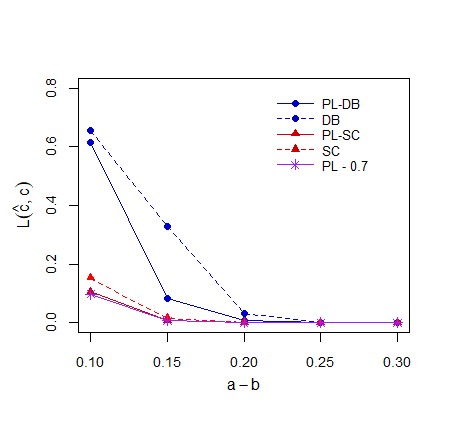}
    \caption{$\pi=(1/3,1/3,1/3)$}
  \end{subfigure}
  \begin{subfigure}[b]{0.48\textwidth}
  \centering
    \includegraphics[scale=0.45]{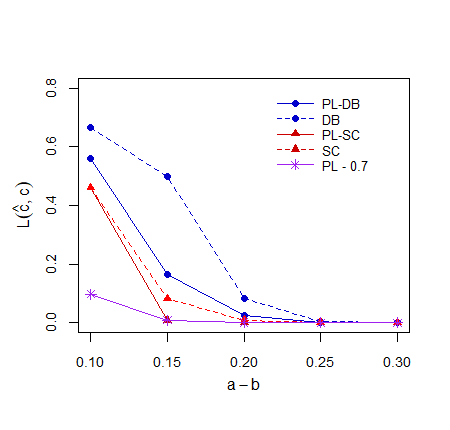}
    \caption{$\pi=(0.2,0.5,0.3)$}
  \end{subfigure}
  \caption{Balanced vs.\ unbalanced community sizes.   Overall error for PL initialized with $\gamma=0.7$, SC, DB, and PL initialized with either SC (PL-SC) or DB (PL-DB).  For both settings, $\sigma^2=0.5$, $n = 1000$, and results are averaged over 100 replications.} 
  \label{fig:pseudo_initial_SC_DB}
\end{figure}

\begin{figure}[!htb]
\centering
  \begin{subfigure}[b]{0.48\textwidth}
 \centering
     \includegraphics[scale=0.5]{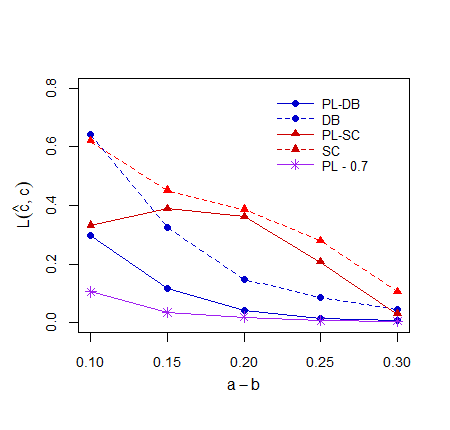}
   \caption{\centering Overall error.}
   \label{subfig:unb2}
  \end{subfigure}
  \begin{subfigure}[b]{0.48\textwidth}
  \centering
  \includegraphics[scale=0.5]{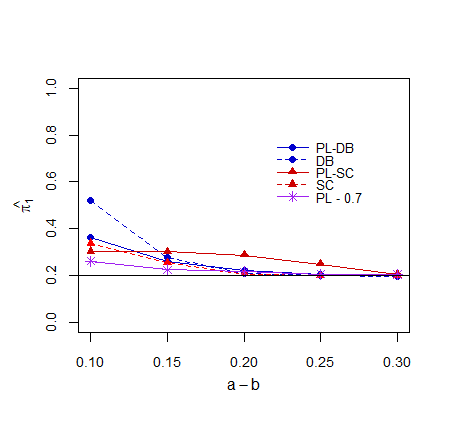}
   \subcaption{\centering Estimated $\hat \pi_1$.}
  \end{subfigure} \\
  \begin{subfigure}[b]{0.48\textwidth}
    \includegraphics[scale=0.5]{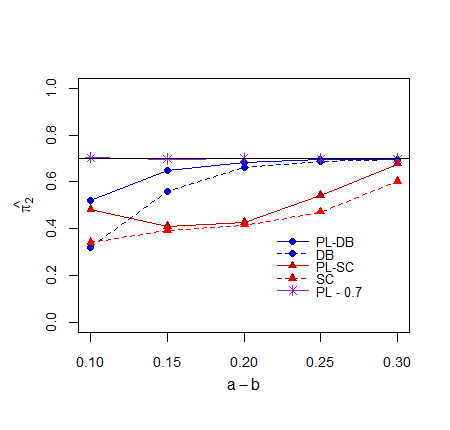}
       \subcaption{\centering Estimated  $\hat \pi_2$. }
      \end{subfigure}
        \begin{subfigure}[b]{0.48\textwidth}
          \includegraphics[scale=0.5]{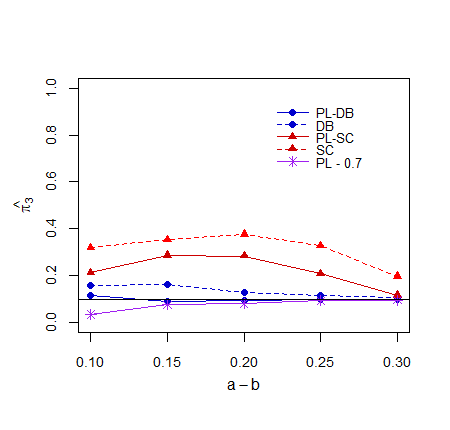}
             \subcaption{\centering Estimated $\hat \pi_3$. }
  \label{fig:pi_hat_subfig}
  \end{subfigure}
    \caption{Unbalanced case, $\pi = (0.2, 0.7, 0.1)$.   Overall error and estimated proportion of nodes   for PL initialized with $\gamma=0.7$, SC, DB, and PL initialized with either SC (PL-SC) or DB (PL-DB). For every setting, $\sigma^2=0.5$, $n=1000$, and results are averaged over 100 replications.}
  \label{fig:pi_hat}
\end{figure}


To investigate robustness to the Gaussian assumption, we also considered the case of heavier-tailed edge weights.   We generated the weights from a mixture of Gaussian and a noncentral $t$ distribution $t_{\mu,d}$ with $d$ degrees of freedom and the noncentrality parameter $\mu$.   The within-community and between-community edge weights are generated by the mixture $\alpha \mathscr{N}(0.2,0.25) + (1-\alpha)t_{0.2,4}$ and $\alpha \mathscr{N}(0,0.25) + (1-\alpha)t_{0,4}$, respectively.  Figure \ref{subfig:gaussian_t} illustrates the within and between densities for $\alpha=0.4$.   Figure \ref{subfig:gaussian_t} shows that the PL method performs well even for small values of $\alpha$, when the edge weights are heavy-tailed. 

Figure \ref{subfig:gaussians} illustrates a different violation of the distributional assumption, with within-community edge weights generated from a bimodal mixture of Gaussians $0.5\mathscr{N}(-0.3,0.25) + 0.5\mathscr{N}(b,0.25)$, and the between-community edge weights are $\mathscr{N}(0,0.25)$.   We vary $b$ from $0.3$ to $0.6$, with the mean of within-community edge weights varying from $0$ to $0.15$, while the between-community mean is always $0$.   As expected, the overall error of the PL estimates decreases as the difference between the within and between means increases.  The DB algorithm has an advantage in this case because the discretization step helps overcome this departure from normality, which was also pointed out by the authors.   

\begin{figure}[htb!]
\centering
  \begin{subfigure}[b]{0.48\textwidth}
 \centering
     \includegraphics[scale=0.45]{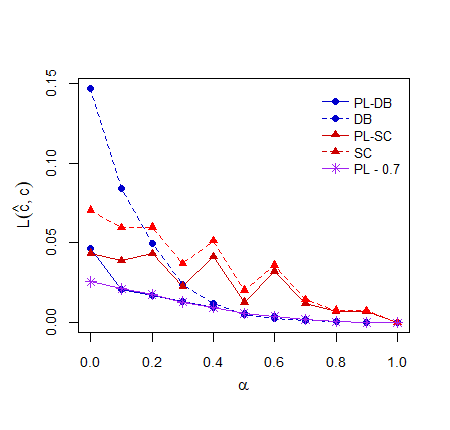}
     \subcaption{Heavy tails}
     \label{subfig:gaussian_t}
  \end{subfigure}
  \begin{subfigure}[b]{0.48\textwidth}
  \centering
  \includegraphics[scale=0.45]{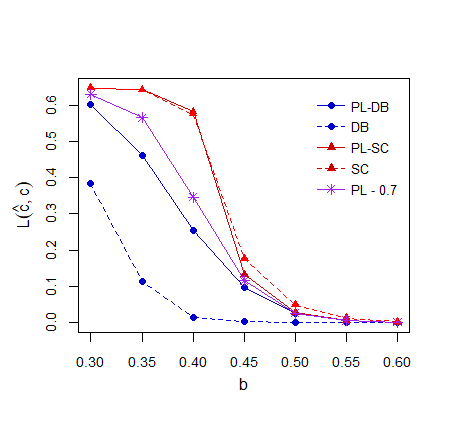}
   \subcaption{Bimodal}
    \label{subfig:gaussians}
  \end{subfigure}
    \caption{Overall error for PL initialized with $\gamma=0.7$, SC, DB, and PL initialized with either SC (PL-SC) or DB (PL-DB). (a)  Edge weights are generated from a mixture of Gaussian and a noncentral $t$-distribution with mixture probability $\alpha$ and $n=1000$. (b) Within-community edge weights are generated from a mixture of Gaussians and between-community weights from a Gaussian distribution. }
  \label{fig:gaussian_t_mixture}
\end{figure}

Finally, we compare the running times of different methods in Figure \ref{fig:time_alg}, on the same standard single core.   The time reported for the PL method does not include the time needed to generate the initial labeling. As $n$ grows,  the PL algorithm becomes cheaper to compute than SC and DB themselves, suggesting it is an effective tool for improving their performance, increasing statistical accuracy at little computational cost.  

\begin{figure}[htb!]
\centering
\begin{subfigure}[b]{0.48\textwidth}
  \centering
    \includegraphics[scale=0.5]{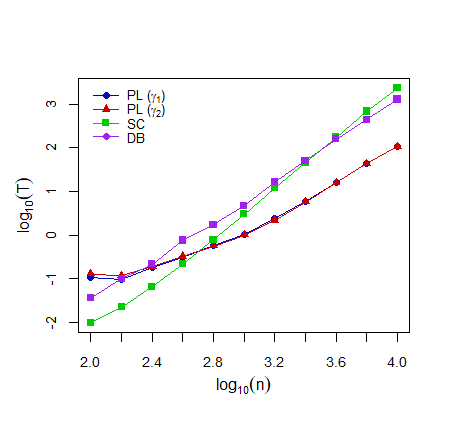}
    \caption{$(a-b)/\sigma=2$; }
    \label{fig:time_alg}
  \end{subfigure}
  \begin{subfigure}[b]{0.48\textwidth}
  \centering
    \includegraphics[scale=0.5]{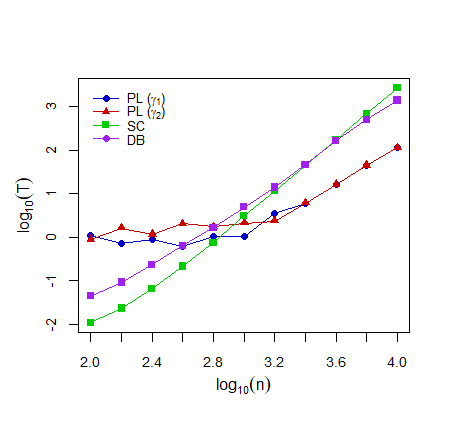}
    \caption{$(a-b)/\sigma=0.2$}
    \label{fig:time_alg_2}
  \end{subfigure}
  \caption{The running time $T$ (in seconds) of the algorithms PL ($\gamma_1=0.9$ and $\gamma_2=0.5$), SC and DB as a function of the number of nodes $n$. }
\end{figure}
\section{Application to fMRI data}\label{sec:data_analysis}

Here we apply the pseudo-likelihood method to the COBRE data \citep{aine2017multimodal}, consisting of resting state fMRI brain images of 54 schizophrenic patients and 69 healthy patients. Each fMRI scan was processed following a standard pipeline at the time of data collection, and converted to a weighted graph with $n=264$ nodes representing the regions of the brain;  see \cite{relion2019network} for details. The edge weights represent functional connectivity between the brain regions, measured by Fisher-transformed correlations between the time series of blood oxygenation levels at the corresponding regions.   Since the Fisher transform of the correlation coefficient is designed to make the distribution approximately normal, this is a natural application for the normally distributed edge weights model.    We average the 69 weighted networks corresponding to healthy patients using  the weighted network average method of \cite{Levin2022recovering} to obtain an estimate for the healthy population, and similarly for the schizophrenic patients, resulting in two ``prototypical'' weighted networks.    

There are no ground truth communities in this problem, but we can still compare the healthy and the schizophrenic populations. We can also compare results from community detection to previously published known brain atlases such as  \cite{power2011functional}.   The true number of communities $K$ is also unknown, and we simply vary $K$ from 2 to 20, a range based on previous findings for fMRI connectivity networks.    As before, we use both DB and SC as potential initial values for our pseudo-likelihood method.  Using the same formula as in the simulation study, the discretization level for DB method is set to 10. 

We start from comparing fitted likelihoods of different starting values and the PL solutions initialized with them, shown in Figure \ref{fig:likelihood_data}.   While we plot these as a function of the number of community $K$ for convenience, the values across different $K$s are not directly comparable, since we are not applying any penalization for model complexity.   The plots generally agree with what we saw in simulations: the PL algorithm finds a better fit than the initial value it starts from, and the DB method especially does not fit the data as well. 

\begin{figure}[htb!]
\centering
  \begin{subfigure}[b]{1\textwidth}
    \includegraphics[scale=0.45]{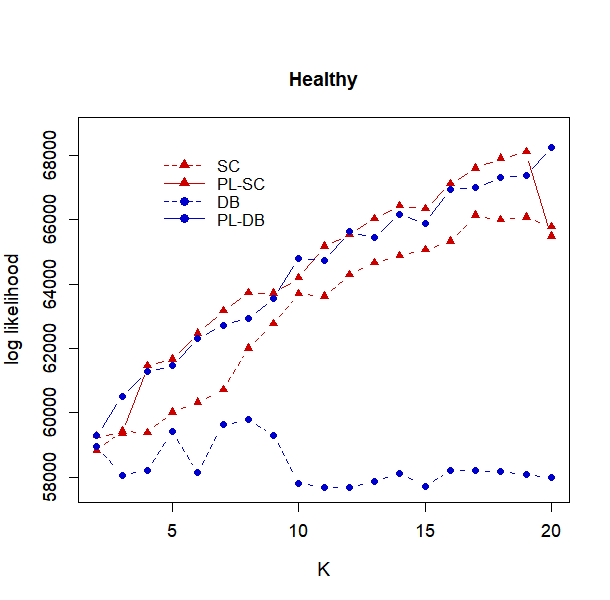}
     \includegraphics[scale=0.45]{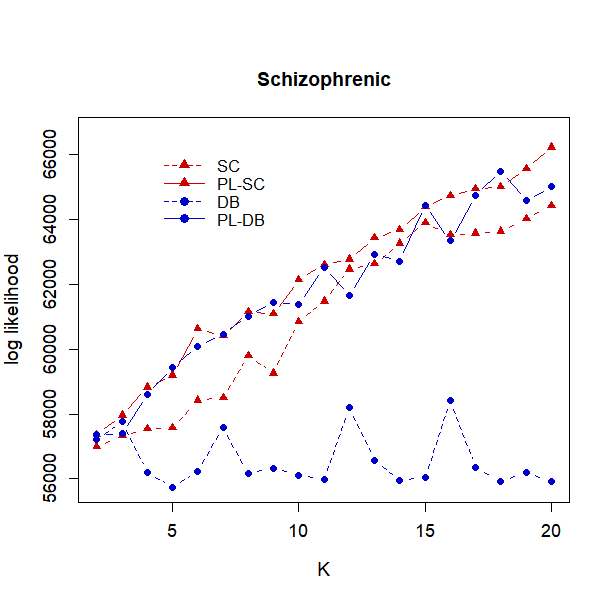}
  \end{subfigure}
  
  \caption{The log of the complete likelihood using the communities from spectral clustering (SC), discretization-based method (DB), and pseudo-likelihood with two initial values, PL-SC and PL-DB.}
  \label{fig:likelihood_data}
\end{figure}

\begin{figure}[htb!]
\centering
    \includegraphics[scale=0.45]{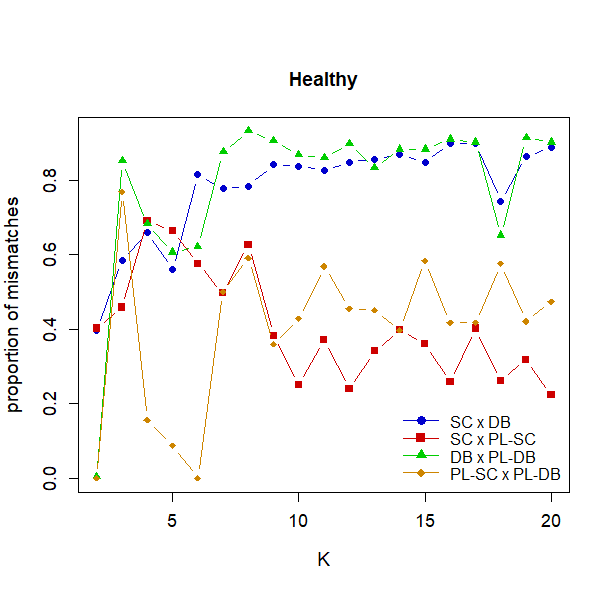}
     \includegraphics[scale=0.45]{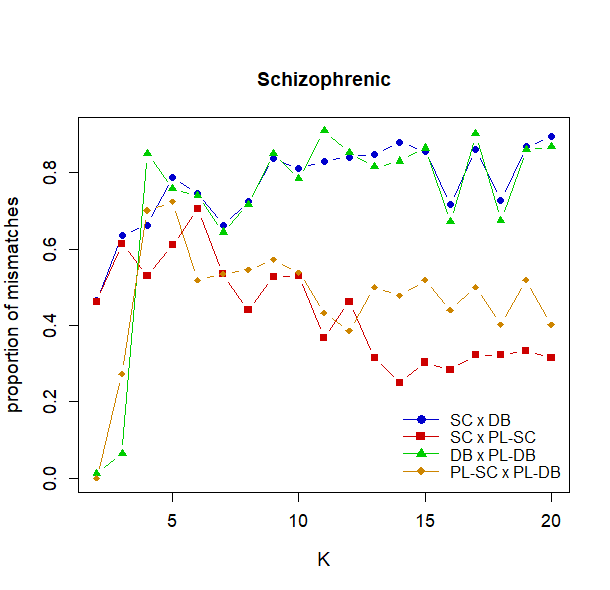}
     \caption{Proportion of mismatched nodes between different pairs of methods as a function of the number of communities $K$.} 
     \label{fig:pseudo_data_a}
\end{figure}

We also looked at how different solutions differ from each other and from their initial values.   Figure \ref{fig:pseudo_data_a} shows the proportion of nodes labeled differently (after finding the best permutation of community labels) between pairs of methods.  We see similar patterns for the healthy and the schizophrenic populations, and while the two initial values are substantially different, the solutions of the PL method are closer, suggesting that it moves in the direction of a higher likelihood from both initial values, giving us more confidence in the PL solutions. 

We also compared our solutions to the Power parcellation \citep{power2011functional}, an assignment of the same $264$ ROIs (nodes)  to 14 functional brain systems (regions), one of which is labeled ``uncertain" and others correspond to known brain functions, as described in Table \ref{tab:parcellation}.     This parcellation was obtained from a different dataset which included only healthy subjects.  

\begin{table}[h]
\caption{Brain regions in the Power parcellation.}
\label{tab:parcellation}
\centering
\begin{tabular}{cp{3.7cm}c|cp{3.3cm}c}
\hline 
Region & Function & Nodes &  Region & Function & Nodes  \\
\hline
P1 & Sensory/somatomotor Hand & 30 & P8 & Fronto-pariental Task Control & 25\\ 
P2 & Sensory/somatomotor Mouth & 5 & P9 & Salience & 18 \\  
P3 & Cingulo-opercular Task Control & 14 & P10 & Subcortical & 13 \\  
P4 & Auditory & 13 & P11 & Ventral attention &9 \\ 
P5 & Default mode & 58 & P12 & Dorsal attention & 11 \\ 
P6 & Memory retrieval & 5 & P13 & Cerebellar & 4 \\ 
P7 & Visual & 31 & P14 & Uncertain & 28 \\ 
\hline 
\end{tabular} 
\end{table}

 Fixing $K=14$ to match the Power parcellation, we estimated 14 communities for both populations by the PL algorithm using SC as the initial value. The estimated communities for the healthy population were relabeled to match their numbers as closely as possible to the Power regions using the Hungarian algorithm \citep{kuhn1955hungarian}.  Table \ref{tab:prop_hps} compares the parcellation estimated from the healthy population (H1-H14) to both the Power regions (P1-P14) and the parcellation estimated for the schizophrenic population (S1-S14), listing the region(s) with the highest overlap for each of the H1-H14 and the proportion of shared nodes.  The Sankey diagram in Figure \ref{fig:sankey_plot} shows the correspondence between the healthy and Power parcellation for reference, and the correspondence between the healthy and the schizophrenic populations.    We see that only a few regions are strongly different between the healthy and the schizophrenic parcellations; in particular, H2 and H11, which both appear to split off from P5, which is the default mode network (DMN).  The DMN has been implicated in schizophrenia previously \citep{brody2009defaulmode,ongur2010defaultmode}, and we have observed in previous work \cite{kim2019graph} that different parts of the DMN seem to play different roles, and so the splitting into multiple parts makes sense.    We also note that the most stable Power regions, which were obtained from different patients, seem to correspond to  communities correspond to the Power regions are the two attention regions, visual,  and fronto-pariental task control, which could indicate these are the strongest communities.  This is, of course, exploratory analysis, and making formal inferences requires further study.

%

\begin{table}[ht]
\caption{Proportion of common nodes of the estimated regions of the healthy network (H1-H14) with the correspondent Power parcellation regions (P1-P14) and the estimated regions of the schizophrenic network (S1-S14).}
\label{tab:prop_hps}
\centering
\begin{tabular}{ccc||ccc}
  \hline
  Region & Region & Region & Region & Region & Region \\ 
  \hline
H1 & P1 (0.70) & S1 (0.59) & H8 & P7 (0.48) & S8 (1.00) \\
H2 & P5 (0.61) & S5 (0.44) & H9 & P8 (0.61)  & S9 (1.00)   \\
H3 & P1,P3 (0.40) & S11 (0.40) & H10 & P10 (0.35) &  S10 (0.67)  \\
H4 & P3,P4,P5 (0.26) & S4 (0.95) &  H11 &P5 (0.65) & S11 (0.59)  \\
H5 & P5 (1.00) & S5 (0.85) & H12 & P12 (0.69) & S12 (1.00)  \\
H6 & P5 (0.62) & S6 (0.44) &  H13 & P5 (0.33) & S13 (0.67)  \\
H7 & P7 (0.69) & S7 (0.61)  &  H14 & P14 (1.00) & S14 (1.00) \\
   \hline
\end{tabular}
\end{table}

\begin{figure}[H]
\centering
    \includegraphics[scale=0.85]{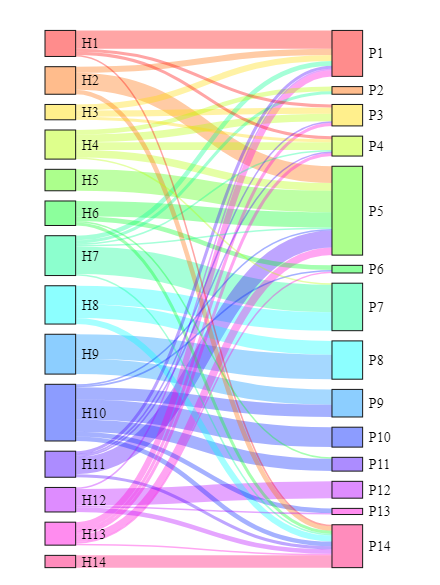}
     \includegraphics[scale=0.85]{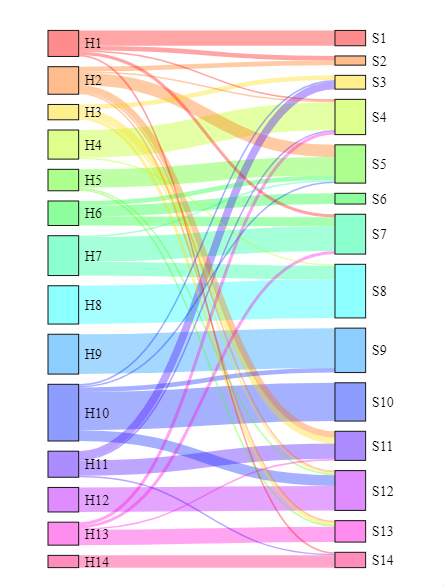}  
  \caption{A Sankey diagram comparing the communities obtained by the PL-SC algorithm for the healthy population (H1-H14) with the Power regions (P1-P14) and the communities estimated for the schizophrenic patients (S1-S14).}
  \label{fig:sankey_plot}
\end{figure}



\section{Discussion}\label{sec:discussion}

In this paper, we proposed a pseudo-likelihood method for community detection on weighted networks, a much less studied type of networks than binary but just as common in practice, or even more common, once we account for the fact that many binary networks are obtained by thresholding weighted networks obtained directly from the data.   Like with many other community detection methods, we provided the analysis under the weighted stochastic block model, but empirically the algorithm appears reasonably robust to model misspecification.  Our theoretical analysis shows that the proposed method achieves the optimal error rate up to a constant, and illustrates the trade-offs between the means, variances, and community sizes.  Like all iterative algorithms, our algorithm depends on how good the initial value is.  We were able to quantify this dependence in theory, and showed empirically that the algorithm is robust to the choice of initial value.  In particular, the initial value provided by spectral clustering is a fast and reliable way to initialize pseudo-likelihood, which can still improve upon it substantially. 

There are many directions in which this work can be taken further.  We only considered theoretical analysis for homogenous models, where all within-communities edge distributions are the same, and between-communities edge distributions are also the same.   This simplification allows for bounds that make the within- and between- tradeoff explicit, which is arguably the main value of obtaining such bounds.   The algorithm itself, however, is equally applicable to heterogeneous models, and theoretical properties under that scenario remain to be investigated.   Another useful extension would  be to incorporate edge distributions with a point mass at zero, so that the network can be truly sparse.   More generally, relaxing parametric assumptions on edge distributions and developing more general versions of the pseudo-likelihood approach to community detection on weighted networks would increase the range of applications where the method could be applied not just based on empirical evidence, but with provable guarantees.  

\section{Acknowledgments}

The authors thank Min Xu, Varun Jog and Po-Ling Loh for sharing their code with us, and Keith Levin for sharing the average networks from the COBRE data.    This research was supported  by the S\~ao Paulo Research Foundation (FAPESP) grant 2018/18115-2  to A.C. and by the NSF grants DMS-1916222, 2052918, and 2210439 to E.L.  This research was conducted while A.C. was a Visiting Postdoctoral Researcher in the Department of Statistics at the University of Michigan, and she thanks the department for its hospitality and support.

\bibliographystyle{apa}
\bibliography{references}

\section{Proofs}

The main idea behind the proof of the Theorem \ref{theo:consitency}  is to bound the probability of misclassification of each node under mild conditions over the parameters estimates $\widehat{a}$, $\widehat{b}$, $\widehat{\sigma}^2$ and the initial labeling $e$. The next lemma states how the labels assigned for each node by the pseudo-likelihood algorithm is related with the block sums defined in \eqref{def:block_sum}.

\begin{lemma}\label{lemma:estimate_label} Assume that $\pi_1=\dots=\pi_K=\frac{1}{K}$. Consider the initial labeling $e\in \mathcal{E}_{\gamma}$ and let $\hat{c}(e)$ be the estimate of the labels obtained by \eqref{eq:EM_estimate}. For any $i\in \{1,\dots,n\}$, $l\in\{1,\dots,K\}$ and initial parameter estimates $\,\widehat{a},\widehat{b},\widehat{\sigma}^2\in S_{a,b}$, we have that

\begin{enumerate}
\item Given $c_i=l$, for $\gamma\in(0,1/K)$, $a<b$  or $\gamma\in(1/K,1)$, $a>b$
\[\widehat{c}_i(e)=l\,,\text{ if and only if }\, s_{il}(e) - s_{ik}(e) >0\,, \text{ for all } k\neq l\,. \]
\item Given $c_i=l$, for $\gamma\in (0,1/K)$, $a>b$  or $\gamma\in(1/K,1)$, $a<b$
\[\widehat{c}_i(e)=l\,,\text{ if and only if }\, s_{il}(e) - s_{ik}(e) <0\,, \text{ for all } k\neq l\,. \]
\end{enumerate}
\end{lemma}

\begin{proof}
Assume that $c_i=l$, for $i \in \{1,\dots,n\}$ and $l \in \{1,\dots,K\}$. For $e\in \mathcal{E}_{\gamma}$, $\gamma \in (0,1)$, the $(r,s)$-th entry of the confusion matrix $R(e)$ is given by
\begin{equation}
R_{rs}=
\begin{cases}
\dfrac{\gamma}{K}, \text{ if } r=s\vspace{0.3cm} \\
\dfrac{(1-\gamma)}{K(K-1)}, \text{ if } r\neq s\,.
\end{cases}
\end{equation}

Let $\widehat{a},\widehat{b}$, $\widehat{\sigma}^2 \in S_{a,b}$ be the estimates of the parameters $a,b,\sigma^2$. Using the symmetry property of $\widehat{B}$, $\widehat{\Sigma}$ and $R$, we have that the $(r,s)$-th entry of $\widehat{P}=n(R(e)\widehat{B})^T$ and $\widehat{\Lambda}=n(R(e)\widehat{\Sigma})^T$ is given by
\begin{equation}\label{eq:P_hat}
\widehat{P}_{rs}=
\begin{cases}
\dfrac{n}{K}(\,\widehat{a}\gamma + \widehat{b}(1-\gamma)\,), \text{ if } r=s\vspace{0.3cm} \\
\dfrac{n}{K(K-1)}(\,\widehat{a}(1-\gamma) + \widehat{b}(K-2+\gamma)\,), \text{ if } r\neq s
\end{cases}
\end{equation}
and
\begin{equation}\label{eq:Lambda_hat}
\widehat{\Lambda}_{rs}= \dfrac{n}{K}\widehat{\sigma}^2\,.
\end{equation}

For any $l \in \{1,\dots, K\}$, we have that
\[\widehat{\pi}_l=\dfrac{1}{n}\sum\limits_{i=1}^n\sum\limits_{k=1}^K \ind\{e_i=l,c_i=k\} = \sum\limits_{k=1}^K R_{lk} = \frac{1}{K}\,.\]

Writing $s_{im}(e)=s_{im}$,  the label update of node $i$, given in \eqref{eq:EM_estimate}, is obtained by
\begin{equation}\label{eq:est_c_hat}
\widehat{c}_i(e) = \arg\min\limits_{k} \left\lbrace \sum\limits_{m=1}^K  \dfrac{(s_{im}-\widehat{P}_{km})^2}{2\widehat{\Lambda}_{km}} + \dfrac{1}{2}\log\widehat{\Lambda}_{km} \right\rbrace\,.
\end{equation}

Using the estimator \eqref{eq:est_c_hat}, we conclude that $\widehat{c}_i(e)=l$, if and only if, for all $k\neq l$
\begin{equation}\label{eq:condition_estimator}
\sum\limits_{m=1}^K  \dfrac{(s_{im}-\widehat{P}_{lm})^2}{2\widehat{\Lambda}_{lm}} - \sum\limits_{m=1}^K  \dfrac{(s_{im}-\widehat{P}_{km})^2}{2\widehat{\Lambda}_{km}}  < \dfrac{1}{2}\sum\limits_{m=1}^K (\log\widehat{\Lambda}_{km} - \log\widehat{\Lambda}_{lm})\,.
\end{equation}

Notice that the RHS of \eqref{eq:condition_estimator} 
is equal to $0$ by \eqref{eq:Lambda_hat}. By symmetry of $\widehat{P}$ , the LRS of \eqref{eq:condition_estimator} can be written as
\begin{equation}\label{eq:eq_squares_p_lambda}
\begin{split}
&\Scale[0.90]{\dfrac{(s_{il}-\widehat{P}_{ll})^2}{2\widehat{\Lambda}_{ll}} - \dfrac{(s_{ik}-\widehat{P}_{kk})^2}{2\widehat{\Lambda}_{kk}} + \dfrac{(s_{ik}-\widehat{P}_{lk})^2}{2\widehat{\Lambda}_{lk}} - \dfrac{(s_{il}-\widehat{P}_{kl})^2}{2\widehat{\Lambda}_{kl}} +\sum\limits_{\substack{m=1 \\ m\neq l, k}}^K  \dfrac{(s_{im}-\widehat{P}_{lm})^2}{2\widehat{\Lambda}_{lm}} - \dfrac{(s_{im}-\widehat{P}_{km})^2}{2\widehat{\Lambda}_{km}}} \\
& \qquad= \dfrac{(s_{il}-\widehat{P}_{ll})^2}{2\widehat{\Lambda}_{ll}} - \dfrac{(s_{ik}-\widehat{P}_{ll})^2}{2\widehat{\Lambda}_{ll}} + \dfrac{(s_{ik}-\widehat{P}_{lk})^2}{2\widehat{\Lambda}_{lk}} - \dfrac{(s_{il}-\widehat{P}_{lk})^2}{2\widehat{\Lambda}_{lk}}\\
&\qquad =  \dfrac{1}{2\widehat{\Lambda}_{ll}}( s_{il} - s_{ik})( s_{il} + s_{ik} - 2\widehat{P}_{ll} ) + \dfrac{1}{2\widehat{\Lambda}_{lk}}( s_{ik} - s_{il})( s_{ik} + s_{il} - 2\widehat{P}_{lk} )\\
& \qquad=  \dfrac{K}{2n\widehat{\sigma}^2}( s_{il} - s_{ik})\left( - 2\widehat{P}_{ll}+ 2\widehat{P}_{lk}  \right)\,.
\end{split} 
\end{equation}

By \eqref{eq:condition_estimator} and \eqref{eq:eq_squares_p_lambda} we conclude that $\widehat{c}_i(e)=l$, if and only if, for all $k\neq l$
\begin{equation}\label{eq:condition_estimator2}
\dfrac{K}{2n\widehat{\sigma}^2}( s_{il} - s_{ik})\left( - 2\widehat{P}_{ll}+ 2\widehat{P}_{lk}  \right) < 0\,.
\end{equation}

Inequality \eqref{eq:condition_estimator2} holds in the following cases:
\begin{enumerate}[label=(\subscript{C}{{\arabic*}})]
\item  $\widehat{P}_{lk} < \widehat{P}_{ll}$ and $s_{il}(e) - s_{ik}(e) >0$
\item $\widehat{P}_{lk} > \widehat{P}_{ll}$ and $s_{il}(e) - s_{ik}(e) <0$\,.
\end{enumerate}

By \eqref{eq:P_hat} and \eqref{eq:Lambda_hat}, we can write
\begin{equation}\label{eq:diff_P_hat}
\begin{split}
\widehat{P}_{lk} - \widehat{P}_{ll} &= \dfrac{n}{K}\left( \widehat{a}\left( \dfrac{1-\gamma K}{K-1} \right) + \widehat{b}\left( \dfrac{1+\gamma K}{K-1} \right) \right)\\
&= \dfrac{n}{K(K-1)} (\widehat{a} - \widehat{b})(1-\gamma K)\,.
\end{split}
\end{equation}

In this case, we have that
\begin{enumerate}[label=(\subscript{C}{{\arabic*}})]
 \setcounter{enumi}{2}
\item $\widehat{P}_{lk} < \widehat{P}_{ll}$\, if\, $\widehat{a} < \widehat{b}$, $\gamma \in (0,1/K)$ or $\widehat{a} > \widehat{b}$, $\gamma \in (1/K,1)\,.$
\item $\widehat{P}_{lk} > \widehat{P}_{ll}$\, if\, $\widehat{a} > \widehat{b}$, $\gamma \in (0,1/K)$ or $\widehat{a} < \widehat{b}$, $\gamma \in (1/K,1)\,.$
\end{enumerate}

By the definition of the set $S_{a,b}$ in \eqref{eq:assumption_parameters} we have that  if $a < b$ then $\widehat{a} < \widehat{b}$. Thus, combining $(C_1)$, $(C_2)$, $(C_3)$ and $(C_4)$ the result follows.
\end{proof}

The next proposition gives an upper bound on the probability of misclassification of each node in the network. This proposition shows that the upper bound does not depend on the initial parameter estimates $\widehat{a},\widehat{b}$ and $\widehat{\sigma}^2$.

\begin{proposition}\label{prop:mis_prob_balanced}
Assume that $\pi_1=\dots=\pi_K=\frac{1}{K}$.
Consider the initial labeling $e\in \mathcal{E}_{\gamma}$ and let $\hat{c}(e)$ be the estimate of the labels obtained by \eqref{eq:EM_estimate}.
For any node $i\in \{1,\dots,n\}$, $a,b\in \mathbb{R}$, $a\neq b$, $\sigma^2>0$, $\gamma\in (0,1)$, $\gamma\neq\frac{1}{K}$, we have that
\begin{equation}
\mathbb{P}(\,\widehat{c}_i(e)\neq c_i\,) \leq  (K-1)\exp\left\lbrace -\dfrac{1}{4} \dfrac{(\gamma K-1)^2}{K(K-1)^2}\dfrac{n(a-b)^2}{\sigma^2}\right\rbrace\,.
\end{equation}
\end{proposition}

\begin{proof}
Without loss of generality we assume that $c_i=l$, for $i\in\{1,\dots,n\}$ and $l\in \{1,\dots,K\}$. Conditioning on $c_i=l$, we have that
\begin{equation}
S_{il}(e) - S_{ik}(e) \sim \mathscr{N}(P_{ll}-P_{lk}\,,\,\Lambda_{ll}+\Lambda_{lk})\,,
\end{equation}
for all $k\in\{1,\dots,K\}$ and $k\neq l$.

The matrix ${\Lambda}=n(R(e){\Sigma})^T$ can be obtained as in \eqref{eq:Lambda_hat} using the entries $\sigma^2$ of the matrix $\Sigma$. Thus, 
\begin{equation}\label{eq:diff_Lambda}
\Lambda_{ll}+\Lambda_{lk} \,=\, \dfrac{2n\sigma^2}{K}\,.
\end{equation}

In the same way, the entries of the matrix ${P}=n(R(e){B})^T$  can be obtained as in \eqref{eq:diff_P_hat} and we get
\begin{equation}\label{eq:diff_P}
P_{ll}-P_{lk} = \dfrac{n}{K(K-1)} ({a} - {b})(\gamma K-1)\,.
\end{equation}

Using the tail bound of Gaussian random variables \citep[p.~22]{boucheron2013concentration} with mean given by $\eqref{eq:diff_P}$ and variance given by $\eqref{eq:diff_Lambda}$, we obtain
\begin{equation}\label{eq:concentration_1}
\mathbb{P}\left( S_{il}(e) - S_{ik}(e) \geq n(a-b)\left( \dfrac{(\gamma K-1)}{K(K-1)} \right) + t \right) \leq \exp\left\lbrace -\dfrac{Kt^2}{4n\sigma^2}\right\rbrace\,,\quad \text{ if } \quad t\geq 0
\end{equation}
and
\begin{equation}\label{eq:concentration_2}
\mathbb{P}\left( S_{il}(e) - S_{ik}(e) \leq n(a-b)\left( \dfrac{(\gamma K-1)}{K(K-1)} \right) - t \right) \leq \exp\left\lbrace -\dfrac{Kt^2}{4n\sigma^2}\right\rbrace\,,\quad \text{ if } \quad t> 0\,.
\end{equation}\\[-5mm]

Notice that the upper bounds \eqref{eq:concentration_1} and \eqref{eq:concentration_2} do not depend on node $i$ and labels $l$ and $k$. Thus, for $\gamma\in(0,1/K)$, $a<b$ or $\gamma\in(1/K,1)$, $a>b$, applying Lemma \ref{lemma:estimate_label}, the union bound inequality and the concentration inequality \eqref{eq:concentration_2} with $t=n(a-b)\Scale[0.80]{\left( \dfrac{(\gamma K-1)}{K(K-1)} \right)}>0$, we get
\begin{equation}\label{eq:upper_bound_mis_1}
\mathbb{P}(\,\widehat{c}_i(e)\neq l\,) = \mathbb{P}\left( \bigcup_{\substack{k=1\\k\neq l}}^K \{S_{il} - S_{ik} \leq 0\} \right) \leq (K-1)\exp\left\lbrace -\dfrac{1}{4} \dfrac{(\gamma K-1)^2}{K(K-1)^2}\dfrac{n(a-b)^2}{\sigma^2}\right\rbrace\,.
\end{equation}

In the same way, for $\gamma\in (0,1/k)$, $a>b$ and $c>d$ or $\gamma\in(1/K,1)$, $a<b$ and $c<d$, setting $t=-n(a-b)\Scale[0.80]{\left( \dfrac{(\gamma K-1)}{K(K-1)} \right)} >0$, we have that
\begin{equation}\label{eq:upper_bound_mis_2}
\mathbb{P}(\,\widehat{c}_i(e)\neq l\,) = \mathbb{P}\left( \bigcup_{\substack{k=1\\k\neq l}}^K \{S_{il} - S_{ik} > 0\} \right) \leq  (K-1)\exp\left\lbrace -\dfrac{1}{4} \dfrac{(\gamma K-1)^2}{K(K-1)^2}\dfrac{n(a-b)^2}{\sigma^2}\right\rbrace\,.
\end{equation}
The result follows by observing that the upper bounds \eqref{eq:upper_bound_mis_1} and \eqref{eq:upper_bound_mis_2} do not depend on $l$.
\end{proof}

\begin{proof}[Proof of Theorem \ref{theo:consitency}]
The proportion of misclassified nodes given in \eqref{eq:def_loss} can be upper bounded using $\phi$ as the identity permutation function. By Lemma \ref{lemma:estimate_label}, for any initial estimates $\widehat{a},\widehat{b}, \widehat{\sigma}^2 \in S_{a,b}$ the estimated labels $\widehat{c}(e)$ depend only on the block sums using the partition $e$. Using the fact that the upper bound obtained in Proposition \ref{prop:mis_prob_balanced} does not depend on the node $i$, $i=1,\dots,n$, we have that
\begin{equation}
\begin{split}
\mathbb{E}\left[\,L(\widehat{c}\,(e),c)\,\right] \,\leq\, (K-1)\exp\left\lbrace -\dfrac{1}{4} \dfrac{(\gamma K-1)^2}{K(K-1)^2}\dfrac{n(a-b)^2}{\sigma^2}\right\rbrace\,.
\end{split}
\end{equation}

Using the Markov inequality for any $\epsilon > 0$ we conclude that
\begin{equation}
\mathbb{P}\left(\,L(\widehat{c}\,(e),c) > \epsilon\,\right) \leq \dfrac{(K-1)}{\epsilon}\exp\left\lbrace -\dfrac{1}{4} \dfrac{(\gamma K-1)^2}{K(K-1)^2}\dfrac{n(a-b)^2}{\sigma^2}\right\rbrace\,.
\end{equation}

By the union bound inequality we have that\\
\begin{equation}
\mathbb{P}\left(\,\sup\limits_{e\,\in\, \mathcal{E}_{\gamma}}L(\widehat{c}\,(e),c) > \epsilon\,\right) \leq \left|\mathcal{E}_{\gamma}\right|\dfrac{(K-1)}{\epsilon}\exp\left\lbrace -\dfrac{1}{4} \dfrac{(\gamma K-1)^2}{K(K-1)^2}\dfrac{n(a-b)^2}{\sigma^2}\right\rbrace\,.
\end{equation}
\vspace{0.1cm}

The cardinality of $\mathcal{E}_{\gamma}$ is upper bounded by the number of configurations $e\in \{1,\dots,K\}^n$ such that, for each community $k=1,\dots, K$, we have that $\sum\limits_{i=1}^n\ind\{e_i=k,c_i=k\}=\gamma m$. In this way, since $m=n/K$, we have
\begin{equation}
\begin{split}
|\mathcal{E}_{\gamma}| &\leq {m \choose m\gamma}^K(K-1)^{(1-\gamma)mK}\\
&\leq \exp\left[ \, n(\,h(\gamma) + \kappa_{\gamma}(2n/K)+(1-\gamma)\log(K-1\,))\,\right]\,,
\end{split}
\end{equation}
where the last inequality follows from Lemma 6 in the supplementary material of \cite{amini2013pseudo}. The result follows taking $\epsilon=\exp\left\lbrace -\frac{1}{8} \frac{(\gamma K-1)^2}{K(K-1)^2}\frac{n(a-b)^2}{\sigma^2}\right\rbrace$.
\end{proof}

The proof of an upper bound on the probability of misclassification of each node in the case of unbalanced networks is more complex because we can not eliminate the dependence on the initial parameters estimates $\widehat{a},\widehat{b}$ and $\widehat{\sigma}^2$.

\begin{proof}[Proof of Proposition \ref{prop:unb_misclassification}]
Using the estimated matrices $\widehat{B}$ and $\widehat{\Sigma}$, we have that
\begin{equation}\label{eq:unb_matrix_P}
\widehat{P}= n(R(e)\widehat{B})^T = n \left(\begin{array}{cc}
\hat{a}\gamma_1\pi_1 + \hat{b}(1-\gamma_2)\pi_2 & \hat{a}(1-\gamma_1)\pi_1 + \hat{b}\gamma_2\pi_2 \\ 
\hat{b}\gamma_1\pi_1 + \hat{a}(1-\gamma_2)\pi_2 &  \hat{b}(1-\gamma_1)\pi_1 + \hat{a}\gamma_2\pi_2
\end{array}  \right)
\end{equation}
and
\begin{equation}\label{eq:unb_matrix_L}
\widehat{\Lambda}= n(R(e)\widehat{\Sigma})^T = n\widehat{\sigma}^2 \left(\begin{array}{cc}
\widetilde{\pi}_1 & \widetilde{\pi}_2 \\ 
\widetilde{\pi}_1 & \widetilde{\pi}_2
\end{array}  \right)\,.
\end{equation}

Starting the EM algorithm with labels $e\in\mathcal{E}_{\,\gamma_1,\gamma_2}$ and estimated matrices $\widehat{B}$ and $\widehat{\Sigma}$, the algorithm returns after one iteration the label of node $i$ such that
\begin{equation}
\widehat{c}_i(e) = \arg\max\limits_{k=1,2} \left\lbrace \widetilde{\pi}_k - \left( \sum\limits_{m=1}^K  \dfrac{(s_{im}(e)-\widehat{P}_{km})^2}{2\widehat{\Lambda}_{km}} + \dfrac{1}{2}\log\widehat{\Lambda}_{km}\right) \right\rbrace\,.
\end{equation}

Using the estimator defined above and the fact that $\widehat{\Lambda}_{12} = \widehat{\Lambda}_{22}$ and 
$\widehat{\Lambda}_{11} = \widehat{\Lambda}_{21}$, we conclude that $\widehat{c}_i(e)=1$, if and only if,
\begin{equation}\label{eq:unb_condition_estimator}
\sum\limits_{k=1}^2  \dfrac{(s_{ik}(e)-\widehat{P}_{1k})^2}{\widehat{\Lambda}_{1k}} - \sum\limits_{k=1}^2  \dfrac{(s_{ik}(e)-\widehat{P}_{2k})^2}{\widehat{\Lambda}_{2k}}  <  2\log\left(\frac{\widehat{\pi}_1}{\widehat{\pi}_2}\right)\,.
\end{equation}

Computing the difference of squares, we have that
\begin{equation}
\begin{split}
 \dfrac{(s_{i1}(e)-\widehat{P}_{11})^2}{\widehat{\Lambda}_{11}} -  \dfrac{(s_{i1}(e)-\widehat{P}_{21})^2}{\widehat{\Lambda}_{21}} &= \dfrac{1}{n\widehat{\sigma}^2\widehat{\pi}_1}\left[ (s_{i1}(e)-\widehat{P}_{11})^2 -(s_{i1}(e)-\widehat{P}_{21})^2\right]\\
 & = \dfrac{1}{n\widehat{\sigma}^2\widehat{\pi}_1}\left(2s_{i1}(e) - (\widehat{P}_{11}+\widehat{P}_{21})\,\right)\left( \widehat{P}_{21}-\widehat{P}_{11}\right)
\end{split}
\end{equation}
and
\begin{equation}
\begin{split}
 \dfrac{(s_{i2}(e)-\widehat{P}_{12})^2}{\widehat{\Lambda}_{12}} -  \dfrac{(s_{i2}(e)-\widehat{P}_{22})^2}{\widehat{\Lambda}_{22}} 
= \dfrac{1}{n\widehat{\sigma}^2\widehat{\pi}_2}\left(2s_{i2}(e) - (\widehat{P}_{12}+\widehat{P}_{22})\,\right)\left( \widehat{P}_{22}-\widehat{P}_{12}\right)\,.
\end{split}
\end{equation}

After some calculations of the entries of $\widehat{P}$ given in \eqref{eq:unb_matrix_P}, we conclude that
\begin{equation}
\begin{split}
&\widehat{P}_{21}-\widehat{P}_{11} = n(\hat{a} - \hat{b})\left((1-\gamma_2)\pi_2 - \gamma_1\pi_1\right)\\
&\widehat{P}_{22}-\widehat{P}_{12} = n(\hat{a} - \hat{b})\left(\gamma_2\pi_2 - (1-\gamma_1)\pi_1\right)
\end{split}
\end{equation}

By the definition of $\beta_1$ and $\beta_2$ given in \eqref{eq:unb_beta}, the LRS of \eqref{eq:unb_condition_estimator} is given by
\begin{equation}\label{eq:unb_condition_estimator1}
\begin{split}
&\dfrac{(\widehat{a} -\widehat{b})}{\widehat{\sigma}^2\widetilde{\pi}_1\widetilde{\pi}_2}\left[\,\beta_1\left(2s_{i1}(e) - (\widehat{P}_{11}+\widehat{P}_{21})\,\right) - \beta_2\left(2s_{i2}(e) - (\widehat{P}_{12}+\widehat{P}_{22})\,\right)\,\right]\\
& = \dfrac{(\widehat{a} -\widehat{b})}{\widehat{\sigma}^2\widetilde{\pi}_1\widetilde{\pi}_2}\left[\,2\beta_1s_{i1}(e) -2\beta_2s_{i2}(e) - \beta_1\left( \widehat{P}_{11}+\widehat{P}_{21}\right) +\beta_2\left( \widehat{P}_{12}+\widehat{P}_{22}\right) \,\right]\,.
\end{split}
\end{equation}

By \eqref{eq:unb_condition_estimator} and \eqref{eq:unb_condition_estimator1} we conclude that $\widehat{c}_i(e)\neq 1$ if
\begin{equation}
\dfrac{(\widehat{a} -\widehat{b})}{\widehat{\sigma}^2\widetilde{\pi}_1\widetilde{\pi}_2}\left[\,2\beta_1s_{i1}(e) -2\beta_2s_{i2}(e) - \beta_1\left( \widehat{P}_{11}+\widehat{P}_{21}\right) +\beta_2\left( \widehat{P}_{12}+\widehat{P}_{22}\right) \,\right] \geq 2\log\left(\frac{\widetilde{\pi}_1}{\widetilde{\pi}_2}\right)\,,
\end{equation}
and $\widehat{c}_i(e)\neq 2$ if
\begin{equation}
\dfrac{(\widehat{a} -\widehat{b})}{\widehat{\sigma}^2\widetilde{\pi}_1\widetilde{\pi}_2}\left[\,2\beta_1s_{i1}(e) -2\beta_2s_{i2}(e) - \beta_1\left( \widehat{P}_{11}+\widehat{P}_{21}\right) +\beta_2\left( \widehat{P}_{12}+\widehat{P}_{22}\right) \,\right] \leq 2\log\left(\frac{\widetilde{\pi}_1}{\widetilde{\pi}_2}\right)\,.
\end{equation}

Given $c_i=1$, we have that
\begin{equation}
2\beta_1s_{i1}(e) - 2\beta_2s_{i2}(e) \sim \mathscr{N}(2\beta_1P_{11} - 2\beta_2P_{12}\,,\,(2\beta_1)^2\Lambda_{11} + (2\beta_2)^2\Lambda_{12})\,.
\end{equation}

In the same way, given $c_i=2$, we have that
\begin{equation}
2\beta_1s_{i1}(e) - 2\beta_2s_{i2}(e) \sim \mathscr{N}(2\beta_1P_{21} - 2\beta_2P_{22}\,,\,(2\beta_1)^2\Lambda_{21} + (2\beta_2)^2\Lambda_{22})\,.
\end{equation}

Calculating the matrix $\Lambda$ as computed in \eqref{eq:unb_matrix_L} using the matrix $\Sigma$ instead of $\widehat{\Sigma}$, we have
\begin{equation}
\begin{split}
(2\beta_1)^2\Lambda_{11} + (2\beta_2)^2\Lambda_{12} &= (2\beta_1)^2\Lambda_{21} + (2\beta_2)^2\Lambda_{22} \\
&= 4n\sigma^2(\beta_1^2{\pi}_1 + \beta_2^2{\pi}_2)\\
&=4n\sigma^2\tau^2
\end{split}
\end{equation}

Using the tail bound of Gaussian random variables \citep[p.~22]{boucheron2013concentration} we obtain
\begin{equation}\label{eq:unb_concentration_1}
\mathbb{P}\left( 2\beta_1s_{i1}(e) - 2\beta_2s_{i2}(e) > 2\beta_1P_{11} - 2\beta_2P_{12} + t  \,|\, c_i=1\right) \leq \exp\left\lbrace -\dfrac{t^2}{4n\sigma^2\tau^2}\right\rbrace\,,\, t\geq 0
\end{equation}
and
\begin{equation}\label{eq:unb_concentration_2}
\mathbb{P}\left( 2\beta_1s_{i1}(e) - 2\beta_2s_{i2}(e) \leq 2\beta_1P_{21} - 2\beta_2P_{22} - t  \,|\, c_i=2\right) \leq \exp\left\lbrace -\dfrac{t^2}{4n\sigma^2\tau^2}\right\rbrace\,,\, t> 0
\end{equation}

If $\,\widehat{a} > \widehat{b}$ and $t_1= \frac{2\sigma^2\widetilde{\pi}_1\widetilde{\pi}_2}{\widehat{a}-\widehat{b}}\log\left(\frac{\widetilde{\pi}_1}{\widetilde{\pi}_2}\right) +\beta_1\left( \widehat{P}_{11}+\widehat{P}_{21}\right) -\beta_2\left( \widehat{P}_{12}+\widehat{P}_{22}\right)-2\beta_1P_{11} + 2\beta_2P_{12} \geq 0$ we have
\begin{equation}
\mathbb{P}(\widehat{c}_i(e)\neq 1 \,|\, c_i=1) \leq \exp\left\lbrace -\dfrac{t_1^2}{4n\sigma^2\tau^2}\right\rbrace\,.
\end{equation}

Analogously, if $\,\widehat{a} > \widehat{b}$ and 
$$t_2= -\frac{ 2\sigma^2\widetilde{\pi}_1\widetilde{\pi}_2}{\widehat{a}-\widehat{b}}\log\left(\frac{\widetilde{\pi}_1}{\widetilde{\pi}_2}\right) -\beta_1\left( \widehat{P}_{11}+\widehat{P}_{21}\right) +\beta_2\left( \widehat{P}_{12}+\widehat{P}_{22}\right)+2\beta_1P_{21} - 2\beta_2P_{22} > 0$$
 we have
\begin{equation}
\mathbb{P}(\widehat{c}_i(e)\neq 2 \,|\, c_i=2) \leq \exp\left\lbrace -\dfrac{t_2^2}{4n\sigma^2\tau^2}\right\rbrace\,.
\end{equation}

After some calculations using \eqref{eq:unb_matrix_P} and \eqref{eq:unb_beta} to compute $t_1$ and $t_2$ we obtain that $t_1= \frac{2\sigma^2\widetilde{\pi}_1\widetilde{\pi}_2}{\widehat{a}-\widehat{b}}\log\left(\frac{\widetilde{\pi}_1}{\widetilde{\pi}_2}\right) + nF(a,b)$ and $t_2 = -\frac{ 2\sigma^2\widetilde{\pi}_1\widetilde{\pi}_2}{\widehat{a}-\widehat{b}}\log\left(\frac{\widetilde{\pi}_1}{\widetilde{\pi}_2}\right) -nF(b,a)$. The result follows analogously when $\widehat a < \widehat b$.
\end{proof}

\end{document}